\documentclass[aps,pra,twocolumn,floatfix,nofootinbib,superscriptaddress,graphics,longbibliography]{revtex4-2}
\usepackage[utf8]{inputenc}
\usepackage[T1]{fontenc}

\usepackage{bm,graphicx,mathrsfs,amsmath,amssymb,mathtools,makecell,bbm, amsthm,dsfont,color,times,txfonts,nicefrac,framed,float,enumitem,tikz,physics,wrapfig,amsfonts,tcolorbox}
\usepackage{physics}
\usepackage{nicematrix}
\usepackage[colorlinks=true,linkcolor=equationcolor,citecolor=teal]{hyperref}
\usepackage{xcolor}
\usepackage{tensor}

\hypersetup{urlcolor=teal}
\definecolor{equationcolor}{RGB}{222,94,100}
\definecolor{boxcolor}{RGB}{215,215,253}
\definecolor{block2c}{RGB}{215,185,212}
\definecolor{block3c}{RGB}{216,155,172}
\definecolor{block4c}{RGB}{216,125,132}

\definecolor{frenchbeige}{rgb}{0.65, 0.48, 0.36}

\makeatletter
\def\blfootnote{\gdef\@thefnmark{}\@footnotetext}
\makeatother

\usepackage[normalem]{ulem}

\newtheorem{thm}{Theorem}

\newtheorem{prop}[thm]{Proposition}
\newtheorem{cor}[thm]{Corollary}
\newtheorem{obs}[thm]{Observation}
\newtheorem{defn}{Definition}
\newtheorem{const}{Construction}
\newtheorem{note}{Note}

\graphicspath{{Figures/}}

\begin{document}

\title{Quantum pushforward designs}

\author{Jakub Czartowski}
	 \affiliation{Faculty of Physics, Astronomy and Applied Computer Science, Jagiellonian University, 30-348 Kraków, Poland}

 \author{Karol {\.Z}yczkowski}
	\affiliation{Faculty of Physics, Astronomy and Applied Computer Science, Jagiellonian University, 30-348 Kraków, Poland}
    
    \affiliation{Center for Theoretical Physics (CFT), Polish Academy of Sciences,  Al. Lotnik{\'o}w 32/46, 02-668 Warszawa, Poland}
 
\date{Feb. 24, 2025}

\begin{abstract}
    Designs, structures connected to averaging with respect to a given measure using finite sets of points, have proven themselves as invaluable tools across the field of quantum information, finding their uses in state and process tomography, key distribution and others. In this work, we introduce a new concept of pushforward designs, which allows us to obtain new structures from already existing ones by mapping them between the spaces, with specific examples including simplex designs and mixed state designs from complex projective designs. Based on the general concept, we put forward a structure called channel $[t,k]$-design, allowing for averaging over space of quantum channels for systems in contact with an environment of dimension $k$. Based on this notion, we introduce the concept of effective environment dimensionality $k^*$, which we estimate for the IBM Kyoto quantum computer to be below $2.2$ for times up to $350\mu\text{s}$. 
\end{abstract}

\maketitle

\section{Introduction} \label{sec:intro}

The problem of averaging or integrating functions with respect to a given measure using a finite set of points has been present in the mathematical community ever since Gauss introduced the idea of quadratures which can be used to approximate the integral of any given function up to a certain polynomial~\cite{Gauss1815}. This and similar methods have been ubiquitous in the context of numerical integration, which appears in many computational tasks across a plethora of different fields. Similar problems naturally extend to arbitrary measurable spaces, on which it makes sense to define functions, giving rise to the idea of designs \cite{levenshtein1992a}.

The concept of designs has gained novel meaning in the context of quantum information, where they have been adapted for complex projective and unitary spaces. The former have been used for state tomography~\cite{scott2006tight}, quantum key distribution~\cite{klappenecker2005mutually}, or sampling derandomisation~\cite{AE07}, among others. The latter have been shown to be useful for ancilla-assisted process tomography~\cite{scott2008optimizing} and identification of universal quantum gate sets~\cite{sawicki2022universality}. In turn, there has been significant interest in finding exemplars of such structures, including mutually unbiased bases and symmetric informational for complex projective designs and the Clifford group for unitary designs~\cite{durt2010MUBs, webb2016Clifford}.

In recent years there have been works that connect already known designs in spaces such as complex projective space with other spaces, such as space of mixed states~\cite{czartowski2020isoentangled}, probability simplex or projective torus~\cite{iosue2023projective}, by considering what may be termed as \textit{induced} or \textit{pushforward measures}. In face of such developments, this manuscript aims at unifying the aforementioned concepts under a common term of \textit{pushforward designs,} taking inspiration from standard nomenclature used in the context of manifold and differential geometry for objects obtained as images of a certain map from one space to another.

This work is organised as follows. In Section~\ref{sec:prelim} we introduce the basic concept of design and present two specific examples relevant to quantum information -- complex projective designs and unitary designs. Next, in Section~\ref{sec:pushforward} we introduce the concept of pushforward designs, which we exemplify in Section~\ref{sec:simplex} by simplex designs obtained from complex projective designs. In Section~\ref{sec:channel} we introduce the original concept of channel designs and demonstrate that they can be obtained as pushforward designs from unitary designs. Then, in Section~\ref{sec:dimension} we introduce the concept of effective environment dimensionality, which is estimated for \textit{IBM Kyoto}.

\section{Preliminaries} \label{sec:prelim}

\subsection{General concept of $t$-design}    

    \begin{figure}[h]
        \centering
        \includegraphics[width=\linewidth]{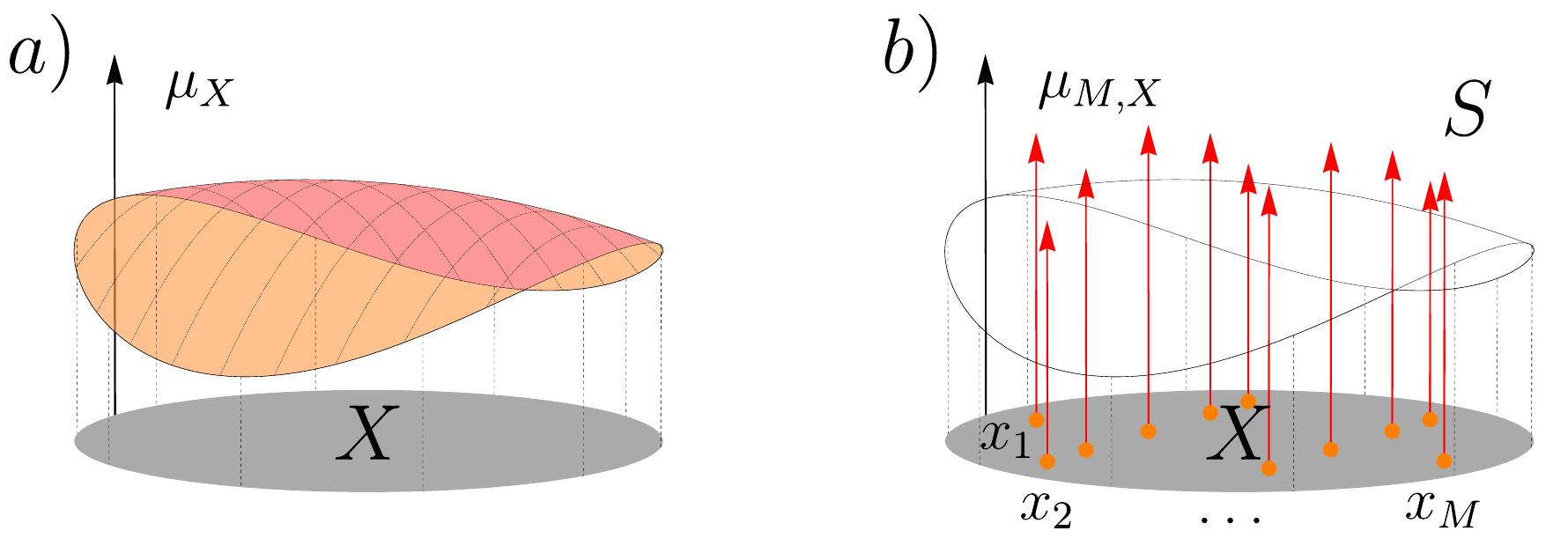}
        \caption{\textbf{Design $S$ as an approximation of a measure:} Consider a space $X\subseteq\mathbb{R}^d$ and a probability measure $\mu_X$ defined on this set [panel a)]. A $t$-design is a set of $M$ points $S = \qty{x_i}_{i=1}^M$, which define a discrete measure $\mu_{M,X}$ mimicking the continuous measure $\mu_X$, so that an average of a polynomial functions up to certain degree $t$ over $\mu_{M,X}$ is equal to that over $\mu_{X}$ [panel b)].}
        \label{fig:measure_design}
    \end{figure}

    We will begin this work by introducing a general idea of a $t$-design, or averaging set, which can be summarised qualitatively as answering the following question: \textit{"How to approximate a continuous distribution over a space of states using a finite set of points?"} -- a concept which is represented visually in Fig.~\ref{fig:measure_design}. More formally, we consider a set $S = \qty{s_i\in X}$ of $\abs{S} = M$ points from a given space $X\subseteq \mathbb{R}^d$ endowed with a nonsingular measure $\mu_{X}$. The set $S$ is called a $t$-design if for every polynomial $f_t(s)$ of degree $t$ the average taken over the set $S$ agrees with an average taken over the measure $\mu_X$,
    \begin{equation}
        \frac{1}{\abs{S}} \sum_{s_i\in S} f_t(s_i) = \int_{X} f_t(s)\dd{\mu_X},
    \end{equation}
    where we assume the measure $\mu_X$ to be normalised, \mbox{$\int_X\dd{\mu_X} = 1$}. 
    This concept, also called averaging set \cite{SEYMOUR1984213}, applies directly to notions such as quadrature rules, simplex $t$-designs and spherical $t$-designs \cite{Gauss1815, Delsarte1977}. An even more general notion of designs in metric spaces has been studied extensively by Levehnstein in a series of papers \cite{levenshtein1992a, levenshtein1998a, levenshtein1998b}.

    The whole concept is understood most easily by saying that a {$t$-design} is an approximation of a continuous distribution using a finite number of points, reproducing it up to $t$-copies or degree-$t$ polynomials.
     The averaging sets and in particular,
    the designs on a sphere, can be considered as  
generalizations of the Simpson integration rule
    and find various applications in quantum theory   
    \cite{Iblisdir2006optimal, Cieslinski2024}.
    In order to get a better understanding of
    the general concept of $t$-designs
     we recall two crucial cases: complex projective designs and unitary designs,
     which will be further used
     in this work.

    \subsection{Complex projective designs}

        Prior to discussing designs used in the context of quantum information, we highlight the limitations of the original concept when extended to subspaces of $\mathbb{C}^d$. Without loss of generality, consider a situation where for all $x \in X \subseteq \mathbb{C}^d$ and $\vb{z} = \mqty(e^{i\phi} & 1 & \hdots & 1), \, \phi \in \mathbb{R}$, it holds that $\tilde{x} = z x \in X$ and $\mu_X(x) = \mu_X(\tilde{x})$. In this case, it follows that  
        \begin{equation}
         \int_X \qty(\prod_{k=1}^t z_{j_k}) \dd{\mu_X} = 0
        \end{equation}
        whenever any of the directions agrees with the first complex axis, $j_k = 1$ for some $k$. The average is correct but provides no additional information about the measure $\dd{\mu}_X$.  
        
        To ensure that $t$-designs in complex spaces more accurately reflect the properties of such spaces, one considers the agreement of averages for balanced polynomials $f_{t,t}(\cdot)$ with equal degrees in both $z_i$ and its conjugate $z^*_i$,  
        \begin{equation}
         \frac{1}{\abs{S}} \sum_{s_i \in S} f_{t,t}(s_i) = \int_{X} f_{t,t}(s) \, \dd{\mu_X}.  
        \end{equation}
        Returning to the example, it can be seen that even the average of a balanced $(1,1)$ monomial, such as $\abs{z_1}$, vanishes identically.

        The above definitions can be easily extended to accommodate \textit{weighted $t$-designs}, represented by sets of pairs $(s_i,w_i\in\mathbb{R})$ or even to continuous measures as opposed to discrete sets.

        With this explanation in hand, we may move on to one of the original examples of designs, the complex-projective $t$-designs.
        Studied extensively by Hoggar and Bannai in a series of works~\cite{Hoggar1982, Hoggar1984, BannaiHoggar1985, Hoggar1989, BannaiHoggar1989, Hoggar1992}, complex projective designs have been later recognised for their utility in contexts such as quantum state tomography and others~\cite{AE07, scott2006tight}.

        The definition of a complex projective $t$-design closely follows the general definitions. 

        \begin{defn}[Complex projective design]
            Consider a set of states $S = \qty{\ket{\psi_i}\in\mathcal{H}_d}$. We say that the set $S$ is a \emph{complex projective $t$-design} if the average over the set for any balanced polynomial $f_{t,t}$ in the components of the state and their conjugates is equal to the average over the whole space $\mathcal{H}_d$ with respect to the measure $\dd{\mu}_{\text{Haar}}$ induced from the flat Haar measure over the unitary group $\mathcal{U}(d)$,
            \begin{equation}
                \frac{1}{\abs{S}}\sum_{\ket{\psi_i}\in S} f_{t,t}\qty(\ket{\psi_i}) = \int_{\mathcal{H}_d} f_{t,t}\qty(\ket{\psi}) \dd{\mu_{\text{Haar}}}.
            \end{equation}
        \end{defn}

        It is interesting to note that there exists an efficient numerical method for verification and generation of complex projective $t$-designs. Given an arbitrary set of normalised states $S = \qty{\ket{\psi_i}\in\mathcal{H}_d}$, the entries in the Gram matrix satisfy the so-called Welch bound,
        \begin{equation}
            \sum_{\ket{\psi_i}\in S}
            \sum_{\ket*{\psi_j}\in S} \abs{\ip{\psi_i}{\psi_j}}^{2t} \geq \frac{\abs{S}^2}{\binom{d+t-1}{t}}.
        \end{equation}
        Note that $\binom{d+t-1}{t}$ is the dimension of the $t$-copy symmetric subspace of $\mathcal{H}_d^{\otimes t}$. For derivation and geometric interpretation of the Welch bound, we refer the reader to~\cite{Datta2012WelchGeo}. 

        Complex projective $2$-designs find one of their uses in state tomography. They have been proven to be optimal for linear non-adaptive tomography, with a particularly simple reconstruction formula~\cite{scott2006tight},
        \begin{equation} \label{eq:proj_des_reco_form}
            \rho = \frac{d(d+1)}{\abs{S}} \sum_{i=1}^{\abs{S}} p_i \op{\psi_i} - \mathbb{I}.
        \end{equation}

        The simplest examples of complex projective $2$-designs include two well-known structures: symmetric informationally complete (SIC) POVMs and mutually unbiased bases (MUB). The former are complex relatives of regular simplices in the real space $\mathbb{R}^d$, and are given by a set of $d^2$ states $\ket{\psi_i}$ of constant overlap,
        \begin{equation}
            \abs{\ip{\psi_i}{\psi_j}}^2 = \frac{1}{d+1}.
        \end{equation}
        Such structures have been conjectured to exist for any dimension $d$ \cite{Zauner2011qdesign}, with numerous analytical and numerical examples found over the years, with the record dimension of \mbox{$d=39604$}~\cite{bengtsson2024sicpovms}. Almost all SIC-POVMs so far are covariant with respect to the Weyl-Heisenberg (WH) group, and thus generated by the action of the WH group on a single fiducial vector, with the only outlier in $d=8$ given by the Hoggar lines, generated by a specific complex reflection group \cite{Hoggar1981}.
        
        A pair of MU bases $\ket{e_i},\,\ket*{f_j}$ in dimension $d$ are defined as bases that are "as different as possible", with constant overlap between states from different bases,
        \begin{equation}
            \abs{\ip{e_i}{f_j}}^2 = \frac{1}{d}
        \end{equation}
        for all $i,j\in\qty{1,\hdots,d}$.
        It has been proven that there exist at most $d+1$ such bases in any given dimension $d$, with construction available for prime-power dimensions $d=p^n$ and the smallest composite dimension $d=6$ remaining as an open problem~\cite{durt2010MUBs}.

    \subsection{Interlude: multi-copy approach}

    Before moving on to the unitary designs, it is interesting to note that for $X\subseteq\mathbb{R}^d$ defining $t$-designs in terms of polynomials of order $t$ is equivalent to a $t$-copy approach, where the averages are taken over $t$-copies of the original space, $X^{\otimes t}$,

        \begin{equation}
        \frac{1}{\abs{S}}\sum_{s_i\in S} s_i^{\otimes t} = \int_{X} s^{\otimes t} \dd{\mu_{X,X^{\otimes t}}}. 
    \end{equation}
    It is important to note that the induced measure $\mu_{X,X^{\otimes t}}$ on $X^{\otimes t}$ is defined as $\mu_{X,X^{\otimes t}}(x^{\otimes t}) = \mu_X(x)$ for all $x\in X$ and $\mu_{X,X^{\otimes t}}(p) = 0$ otherwise. Simply put, the measure is restricted to the symmetrc subspace of $X^{\otimes t}$. Later on, for brevity, we will write $\mu_X$ instead of $\mu_{X,X^{\otimes t}}$ whenever the $t$-copy setting is clear from the context.
    
    Equivalence of the $t$-copy and a complete set of monomials which provide a basis for polynomials $f_t$ is straightforward -- the object $s^{\otimes t}$ contains every possible product of $t$ coordinates in $\mathbb{R}^d$, while the symmetric measure $\mu_X$ provides the required agreement between copies.

    The adaptation to complex setting, $X\subseteq\mathbb{C}^d$, is achieved by considering average agreement for copies of objects together with their conjugates,
    \begin{equation}
        \frac{1}{\abs{S}}\sum_{s_i\in S} \qty(s_i\otimes s_i^*)^{\otimes t} = \int_S \qty(s\otimes s^*)^{\otimes t} \dd{\mu_{X}}.        
    \end{equation}
    Equivalence to the formulation in terms of balanced polynomials follows the same reasoning as for the real case. It is interesting to note that the average is technically evaluated over $t$ copies of the space of linear operators $\mathcal{L}(S)$ with the entire measure restricted to (non-normalised) rank-1 projectors onto~$s\in X$.

    \subsection{Unitary designs}

        Since the unitary group $\mathcal{U}(d)$ is the set of natural operations governing the evolution of pure states of $d$-dimensional quantum systems, it is natural to consider structures allowing effective approximation of this group. This, again, can be handled by introducing the concept of \textit{ unitary designs} \cite{AE07, dankert2009exactunitary}.
        \begin{defn}[Unitary $t$-design]
            Consider a set \mbox{$S = \qty{U_i\in\mathcal{U}(d)}$} of elements of the unitary group $\mathcal{U}(d)$ and a balanced polynomial function $f_{t,t}(U)$ in the components of $U$ and its conjugate $U^*$. The set $S$ is called a unitary $t$ design if the average of $f_{t,t}$ over the set is equal to the average over the entire group $\mathcal{U}_d$ with the natural Haar measure,
            \begin{equation}
                \frac{1}{\abs{S}}\sum_{U_i\in S} f_{t,t}(U) = \int_{\mathcal{U}(d)} f_{t,t}(U)\dd{U}.
            \end{equation}
        \end{defn}

        This is the natural definition, and here it is beneficial to express it in the $t$-copy form, directly linked to the adjoint action,
            \begin{equation}
                \frac{1}{\abs{S}}\sum_{U_i\in S} \qty(U_i\otimes U_i^\dagger)^{\otimes t} = \int_{\mathcal{U}(d)} \qty(U\otimes U^\dagger)^{\otimes t}\dd{U}.
            \end{equation}
        In this formulation, it is clear that the problem is equivalent to finding the $I$ matrix known from the Weingarten calculus, and solutions can be found analytically for $1\leq t\leq d$ and have to be calculated numerically for $t > d$~\cite{Collins_2022}.

        A canonical example of a unitary $3$-design is the Clifford group \cite{webb2016Clifford}, significant in quantum information due to the efficient classical simulation of Clifford circuits on non-magic states \cite{gottesman1998heisenberg, aaronson2004improved}. In order to define it, let us first consider the Pauli group on $n$ qubits, given by operators of the form
        \begin{equation}
            \mathcal{P}_n \equiv \qty{P =i^{l} \bigotimes_{a=1}^n\sigma_x^{j_a}\sigma_z^{k_a}},
        \end{equation}
        with $l\in\qty{0,1,2,3}$, $j_a, k_a\in\qty{0,1}$. It can be generated by taking Pauli gates $\sigma_x$ and $\sigma_z$ acting locally on each qubit.  The number of elements in the group is easy to calculate, equal to $\abs{\mathcal{P}_n} = 4^{n+1}$. With this in hand, we may define the Clifford group on $n$ qubits as the group of unitary operations that map the Pauli group onto itself under adjoint action of the unitary group,
        \begin{equation}
            \mathcal{C}_n = \qty{C: C\in SU(2^n),\,\forall_{P\in\mathcal{P}_n} CPC^\dagger\in\mathcal{P}_n}.
        \end{equation}
        Clifford group is generated by local Hadamard operations,
        \begin{equation}
            H = \frac{1}{\sqrt{2}}\mqty(1&1\\1&-1),
        \end{equation}
        local phase-shifts $S = \operatorname{diag}(1,I)$ and all pairwise $\operatorname{CNOT}_{i\rightarrow j}$ operators. The order of the Clifford group grows much more quickly than that of the underlying Pauli group~\cite{Grier2022classificationof},
        \begin{equation}
            \abs{\mathcal{C}_n} = 2^{n^2+2n}\prod_{j=1}^n\qty(4^j-1).
        \end{equation}
        In the table below, we provide a comparison of cardinalities of Pauli and Clifford groups for small~$n$.
        \begin{table}[H]
            \centering
            \begin{tabular}{|c|c|c|c|c|c|}
            \hline
                $n$ & 1 & 2 & 3 & 4 & 5 \\ \hline
                $\abs{\mathcal{P}_n}$ & 16 & 64 & 256 & 1024 & 4096 \\
                $\abs{\mathcal{C}_n}$ & 24 & 11 520 & 92 897 280 & 12 128 668 876 800 & $\sim2.54\cdot10^{19}$ \\
            \hline
            \end{tabular}
            \caption{Orders of the Pauli group $\mathcal{P}_n$ and Clifford group $\mathcal{C}_n$ for small numbers of qubit $n$. Notice rapid growth of the Clifford group, compared with the Pauli group~\cite{Grier2022classificationof}.}
            \label{tab:my_label}
        \end{table}

        It is relatively simple to show the following relation between unitary and complex-projective designs.

        \begin{obs}\label{obs:proj_from_unit}
            Consider a unitary $t$-design $S = \qty{U_i\in\mathcal{U}(d)}$ 
            and an arbitrary pure state $\ket{\psi}\in\mathcal{H}_d$. Then, the set $S\ket{\psi} = \qty{U_i\ket{\psi}}$ is a complex projective $t$-design. 
        \end{obs}
        The proof proceeds straightforwardly by demonstrating the saturation of the Welch bound for the given $t$. Alternatively, the following bound for unitary $t$-designs of a specified $t$, which is equivalent to the Welch bound 
        may be checked. 
        \begin{cor}
            Consider a set of unitary matrices $S = \qty{U_i\in\mathcal{U}(d)}$. The following bound
            \begin{widetext}
            \begin{equation}
                \sum_{a=1}^d\sum_{U_i,\,U_j\in S}\abs{\ev{U_i^\dagger U_j}{a}}^{2t} \!\!\!\!= 
                \sum_{a=1}^d\sum_{U_i,\,U_j\in S}\abs{\sum_{b=1}^d \overline{(U_i)_{ab}}(U_j)_{ba}}^{2t} \!\!\!\!\geq d \abs{S}^2 \binom{d+t-1}{t}^{-1}
            \end{equation}
            \end{widetext}
            is satisfied by any set $S$ and is saturated if and only if $S$ is a unitary $t$-design. 
        \end{cor}

        This, in fact, provides an amenable way to generate unitary designs numerically and to check whether a given set of unitary operations indeed provides a unitary $t$-design.

        In fact, induction of complex projective designs by acting with unitary design in a given fixed state provides an example of the central notion of this work: \textit{pushforward design}.

\section{Pushforward designs} \label{sec:pushforward}

    \begin{figure}[h]
        \centering
        \includegraphics[width=\linewidth]{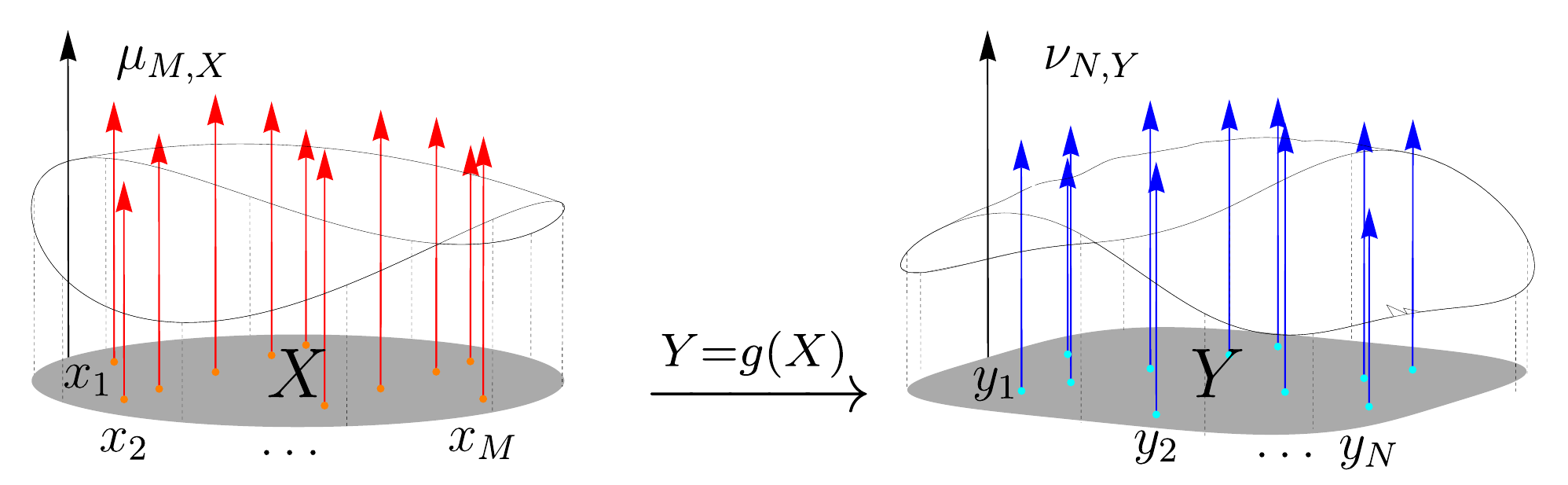}
        \caption{\textbf{Pushforward design}: Let us take $t$-design from one state space $X$ to another state space $Y$ via a mapping function $g: X \to Y$. On the left, the state space $X$ with points $\qty{x_i}_{i=1}^M$ and their associated measure $\mu_{M,X}$ (red arrows). The function $g$ maps these points to the state space $Y$ on the right, resulting in points $\qty{y_j}_{j=1}^N$ with the measure $\nu_{N,Y}$ (blue arrows). 
        }
        \label{fig:pushforward_design}
    \end{figure}

    In order to unify different concepts mentioned in the introduction, such as mixed states designs, complex projective designs obtained from unitary designs, or simplex designs obtained from complex projective designs using maps from one space to another, we introduce the concept of \textit{pushforward design.}

    \begin{defn}[Pushforward design]\label{def:pushforward_design}
        Consider a pair of state spaces $X$ and $Y$, a measurable map $g: X \to Y$, and a non-singular measure $\mu_{M,X}$ on $X$. The function $g$ therefore induces a measure $\nu_{N,Y}$ on $Y$, given by $\nu_{N,Y}(O_Y) = \mu_{M,X}\qty[g^{-1}(O_Y)]$ for any $O_Y\subset Y$. Let $\qty{x_i}_{i=1}^M$ be a $t$-design on $X$. We say that the set $\qty{y_j}_{j=1}^N = g(\qty{x_i})$, with associated weights $w_j$ corresponding to the multiplicities of each $y_j$, is an induced $t$-design if it provides a $t$-design with respect to the pushforward measure $\nu$ on $Y$.
    \end{defn}

    The concept of pushforward design has been depicted in an abstract manner in Fig.~\ref{fig:pushforward_design}. One may instantly formulate the following proposition:
    \begin{prop}[Linear induction of pushforward designs]\label{prop:linear_proj}
        The Definition~\ref{def:pushforward_design} is realised whenever $g(\cdot)$ is linear and non-degenerate. When $\operatorname{dim}(X) > \operatorname{dim}(Y)$, the degree $t'$ of the resulting design may be greater than the pre-image design, $t'\geq t$, and the number of points may be reduced at the cost of degeneracy, resulting in a weighted design. 
    \end{prop}
    \begin{proof}
    Linearity of the function $g$ ensures that the linear functions in $X$ are mapped to the linear functions in $Y$. Then, for equal dimensions $\operatorname{dim}(X) = \operatorname{dim}(Y)$, we find that a complete basis on $X$ is mapped to a complete basis on $Y$, and the agreement between averages with respect to the image measure is guaranteed by the linearity of $g$. Otherwise, when the dimension of the target space $Y$ is less than the dimension of the original space $X$, $\operatorname{dim}(X) > \operatorname{dim}(Y)$, the basis of linear functions on $X$ is mapped onto an over-complete frame on $Y$, thus preserving agreement of averages. The analogous argument applies when applied to $t$-copy averages, since $g^{\otimes t}:X^{\otimes t}\mapsto Y^{\otimes t}$ will remain linear, with the same arguments applied to the dimensionality of symmetric subspaces.

    The increase in degree $t'$ from the original degree $t$ can be attributed to symmetries of the image measure, e.g., all odd-degree polynomials over the $[-1,1]$ interval vanish, and thus any symmetric $2n$-design is automatically a $(2n+1)$-design and, in turn, may result from a properly orientated $2n$-design in the original space $X$. However, in general, such a possibility will depend on the specifics of $X$ and $Y$ spaces and the underlying measures.
    \end{proof}

    \smallskip

    Although general, Proposition~\ref{prop:linear_proj} can be applied to recover already known examples, with the elementary ones including complex projective designs obtained as pushforward designs from unitary designs (see Observation~\ref{obs:proj_from_unit}), mixed state designs obtained from complex projective or -- to go even simpler -- simplex designs in dimension $d$ induced by marginalisation of distribution from simplex designs of product dimension $d\cdot d'$. Before moving forward, however, let us remark that when $\operatorname{dim}(X) \leq \operatorname{dim}(Y)$, the image measure $\nu_Y$ will be restricted to a lower-dimensional hypersurface in $Y$ and zero otherwise. Thus the problem will be effectively turned back to the case,
    in which the dimensions of both spaces
    are equal.
    
    In the following section, we will consider two further constructions, starting with the already known construction of simplex designs by decoherence of complex projective designs~\cite{czartowski2020isoentangled}. Next, we will proceed with a novel idea of channel designs.

    \section{Simplex designs from projective designs} \label{sec:simplex}

    We will begin the exhibition of the concept of pushforward designs by revisiting a relation between the Hilbert space~$\mathcal{H}_d$ and the $d$-point probability simplex $\Delta_d\in\mathbb{R}^{d-1}$
    discussed in~\cite{CGGZ20}.

    A projection from a state $\ket{\psi}\in\mathcal{H}_d$ to the probability simplex is achieved by decoherence of the corresponding density operator with respect to a selected basis, $\ket{\psi}\overset{P}{\mapsto} p_i = \Tr(\op{i}\op{\psi}) = \abs{\ip{\psi}{i}}^2$, which is interpreted as the Born rule governing the probabilities of measuring an input labelled $i$ in a given state $\ket{\psi}$. It is well known that by taking the flat measure in $\mathcal{H}_d$ generated from the Haar measure in unitary operations $U(d)$ induces the flat Lebesgue measure on $d$-point probability simplex $\Delta_d$ in $\mathbb{R}^{d-1}$~\cite{czartowski2020isoentangled}.

    \begin{figure}[h]
        \centering
        \includegraphics[width =  \linewidth]{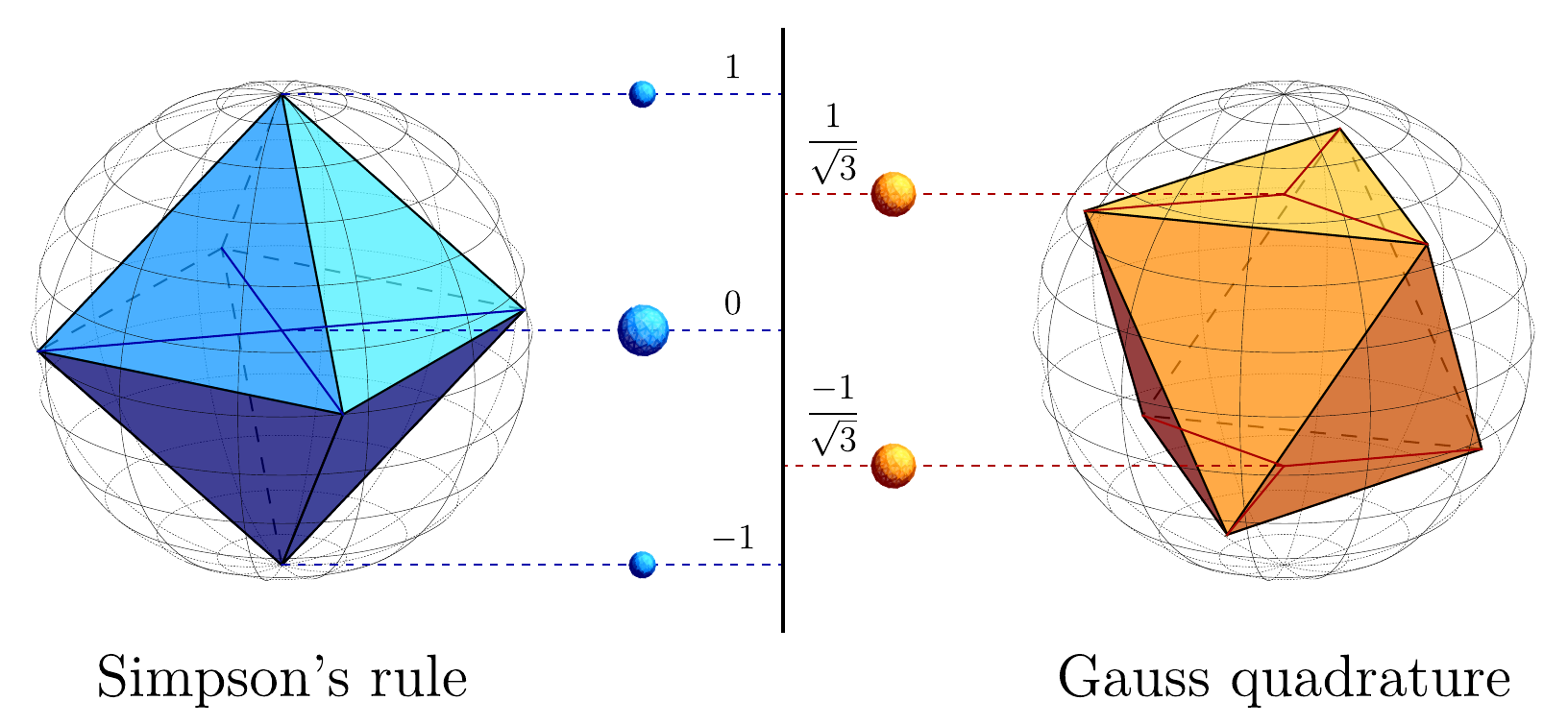}
        \caption{\textbf{Interval designs as projections of spherical designs:} 
        An octahedron defines a spherical 3-design,
        equivalently to $\mathbb{C}P^1$ $3$-design. 
        Depending on its orientation, by projecting it on a given axis one can generate either $1:4:1$ Simpson's rule, or two-point Gauss quadrature.
        }
        \label{fig:line_sphere_des}
    \end{figure}
    
    It may seem that the projection in question is not linear and as such does not fall within the scope of Proposition~\ref{prop:linear_proj}; however, it is worth noting that the map is linear in components of the state and their conjugates, and in turn, we can conclude that any $\mathbb{C}P^{d-1}$ design induces, by projection $P$, a design in~$\Delta_d$.
    
    This relation provides an accessible method for generating designs in the probability simplex $\Delta_d$. Let us consider the Hilbert space $\mathcal{H}_2$, corresponding to the Bloch sphere, and take an exemplary 2-design given by vertices of an octahedron. By considering its projections, as depicted in Fig.~\ref{fig:line_sphere_des}, one obtains interval designs corresponding to well-known integration rules over the $[-1,1]$ interval, the Gauss quadrature, and Simpson rule, related by their common origin from a single underlying complex projective design. Indeed, for any axis, corresponding to a unitary rotation $U$ of the entire complex projective $t$-design, one obtains a different $t$-design in the simplex. However, it is important to note that the minimal arrangements -- $2n$ uniformly weighed points that form the $2n+1$ design over the interval -- are unique.
    
    Dimension $d=3$ is interesting in its own rights, as there exists an entire one-parameter family of non-equivalent SIC-POVMs generated by the action of the Weyl-Heisenberg group on fiducial vectors of the form
    \begin{equation} \label{eq:fid_sic_3d}
        \ket{\psi(\theta)} = \sin \theta\frac{1}{\sqrt{2}}\mqty(0 \\1 \\-1) + \cos\theta \frac{1}{\sqrt{6}}\mqty (2 \\ -1 \\ -1).
    \end{equation}
    Decoherence of these states with respect to the computational basis yields a simplex rescaled by a factor of $1/2$ and rotated by the angle $\theta$. Another example comes from the standard form of MUB, composed of the computational basis and three unitary matrices with entries of constant absolute value, also referred to as \textit{Hadamard matrices}. The configurations originating from both SIC-POVMs and MUB are depicted in~Fig.~\ref{fig:2d_simplex_des}. 

    \begin{figure}[h]
        \centering
        \includegraphics[width=.5\linewidth]{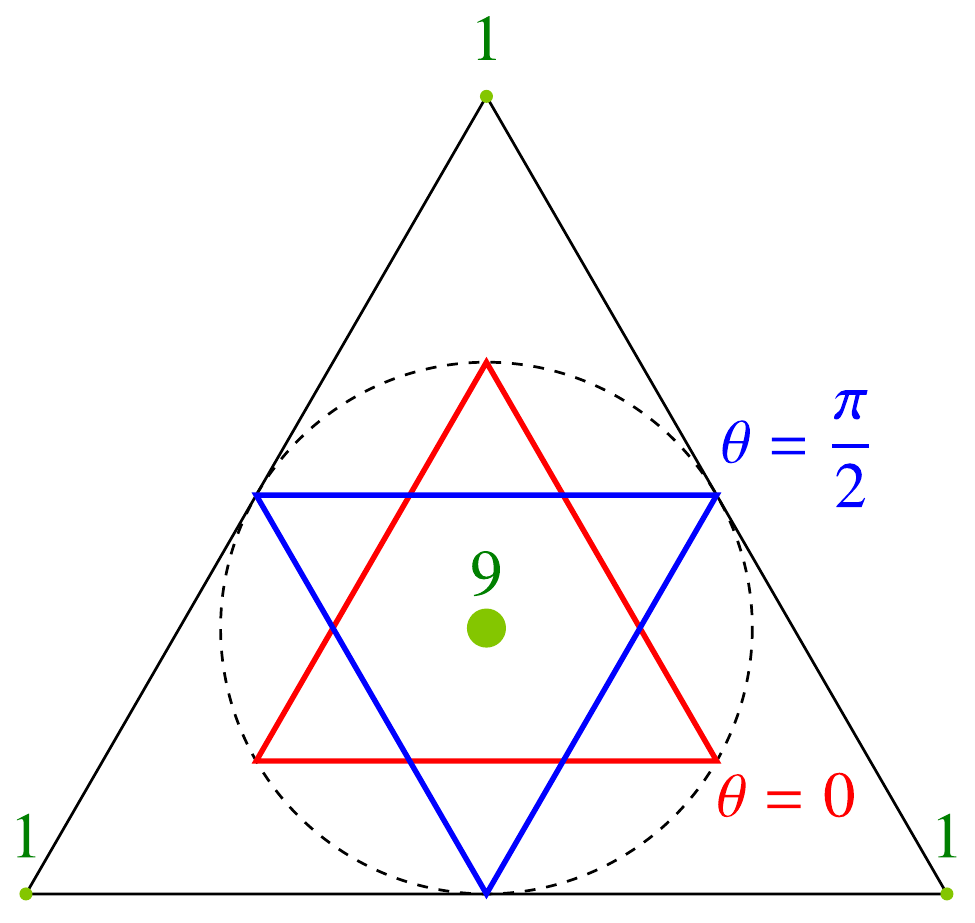}
        \caption{\textbf{2-designs in $\Delta_3$:} Depending on the choice of $\theta$ in the fiducial vector given in Eq.~\eqref{eq:fid_sic_3d} SIC-POVM yields a different simplex 2-design. Note that for every value of $\theta$ one finds an equilateral triangle inscribed in a circle of radius $r = 1/2$. A different configuration with ratios $1:9:1:1$, akin to Simpson rule, is generated by projecting standard set of MUB in $d=3$.}
        \label{fig:2d_simplex_des}
    \end{figure}

    Since dimension $d=4$ is the smallest product dimension, it is instructive to consider simplex 2-designs that can be generated by decoherence of the full set of MUB with respect to different distinguished bases. In standard form, MUB are composed of computational basis and $4$ other bases, which can be represented by Hadamard matrices, thus their decoherence generates points of weight $1$ for all vertices of the probability simplex and a central point of weight $16$. Another distinguished form of MUB in $d=4$ is the isoentangled configuration, in which all states possess the same degree of entanglement~\cite{czartowski2020isoentangled}. In this case, we find a highly nontrivial configuration of $9$ points within the simplex, with all but one of them of equal weight. Thus, one may ask about the existence of \textit{isocoherent MUB} yielding minimal number of points in $\Delta_4$, with all classical states equivalent up to permutation of entries. Existence of such a configuration is guaranteed by the possibility of generating all MUB as powers of a single unitary matrix $U$ -- it is then sufficient to shift to the eigenbasis of the said matrix \cite{seyfarth2011cyclicMUB}. In Appendix \ref{app:isocoh_MUB} we present explicit form of isocoherent MUB, found independently using simple numerical search. All three simplex designs resulting from above forms of MUB are presented in Fig.~\ref{fig:3d_simplex_designs}.   

    \begin{figure}[h]
        \centering
        \includegraphics[width=.9\linewidth]{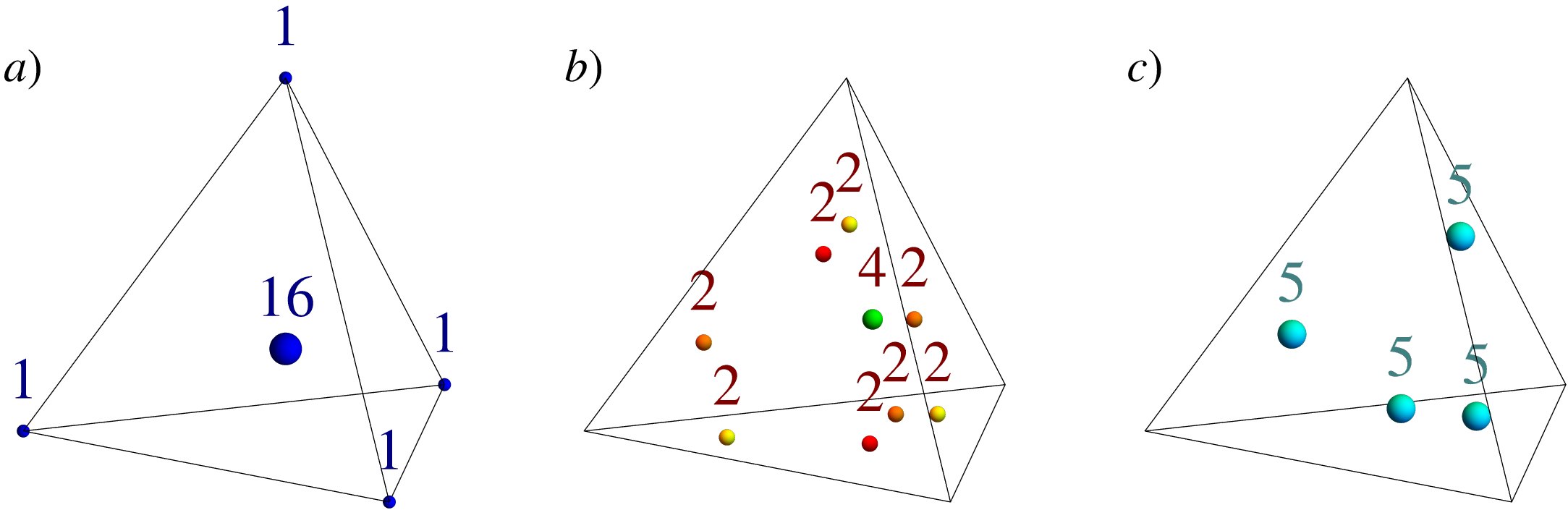}
        \caption{\textbf{2-designs in $\Delta_4$:} Three different configurations of points within the $\Delta_4$ simplex obtained by projecting MUB in $d=4$, with a) the standard form, yielding the multidimensional Simpson rule~\eqref{eq:simpson_rule}, b) irregular arrangement generated by isoentangled MUB and c) numerically found isocoherent MUB, given in Appendix \ref{app:isocoh_MUB}.}
        \label{fig:3d_simplex_designs}
    \end{figure}

    In higher dimensions the situation becomes more complicated, as cyclic permutation of amplitudes generated by the shift operator from the Weyl-Heisenberg group does not necessarily create regular simplices, which is already not the case for $d=4$. Despite this, since each SIC-POVM is a complex projective {2-design}, however seemingly asymmetric the configuration of resulting $d$ points in the simplex is, they are guaranteed to form simplex 2-designs by Proposition~\ref{prop:linear_proj}. Furthermore, based on decoherence of known projective designs, including MUB,  we may put forward the existence of a multidimensional Simpson rule.

    \begin{prop}[Generalized Simpson Rule]
        Let $\Delta_d$ be a $d$-point probability simplex in $R^{d-1}$
        and let $S$ be the set of $d+1$ weighted points in $\Delta_d$ consisting of $d$ points at the vertices with weight one and the central point with weight $d^2$, with overall ratio of
        \begin{equation}
            1:d^2:\underbrace{1:\hdots:1}_{d-1}.   \label{eq:simpson_rule} 
        \end{equation} 
        Then, the configuration given by $S$ forms a $2$-design with respect to the flat Lebesgue measure in $\Delta_d$.
    \end{prop}
    \begin{proof}
        In every dimension $d$ where a complete set of $d+1$ mutually unbiased bases (MUB) exists, such a set forms a complex projective $2$-design~\cite{scott2006tight}. By applying a global rotation, it is always possible to align the first basis with the computational basis, while the remaining $d$ bases correspond to a complete set of mutually unbiased Hadamard matrices. Decoherence of the computational basis
        produces $d$ points in the corners of the 
        $d$-point simplex (of dimension $d-1$),
        while all states forming $d$ Hadamard matrices
        decohere to a point at the center of the simplex with weight $d^2$
        as described in the proposition.

        In dimensions where the problem of existence of complete sets of MUB remains unresolved, one may reason that, if such sets exist, they would similarly yield a simplex $2$-design. Alternatively, it is possible to consider a recently introduced configuration consisting of the computational basis and additional $d^2$ states, which do not form $d$ complete bases but nonetheless satisfy the properties of a complex projective $2$-design~\cite{iosue2023projective}.
        
        Thus, the set $S$ always satisfies the conditions for being a $2$-design with respect to the flat Lebesgue measure on $\Delta_d$.
    \end{proof}

    It is interesting to note that such configurations can be used to approximate averages of functions over arbitrary smooth hypersurfaces in a manner similar to that in which Gauss quadratures are used for numerical integration. To understand this, consider that any such hypersurface can be approximated using triangulation. Next, we note that affine transformations of a simplex applied to the uniform Lebesgue measure generate another uniform measure; thus, also the affine transformation of a simplex $t$-design remains a $t$-design for the new, distorted simplex. Thus, we may use it repeatedly, generating an averaging scheme as presented in Fig.~\ref{fig:mesh_averaging}

    \begin{figure}[h!]
        \centering
        \includegraphics[width=.8\linewidth]{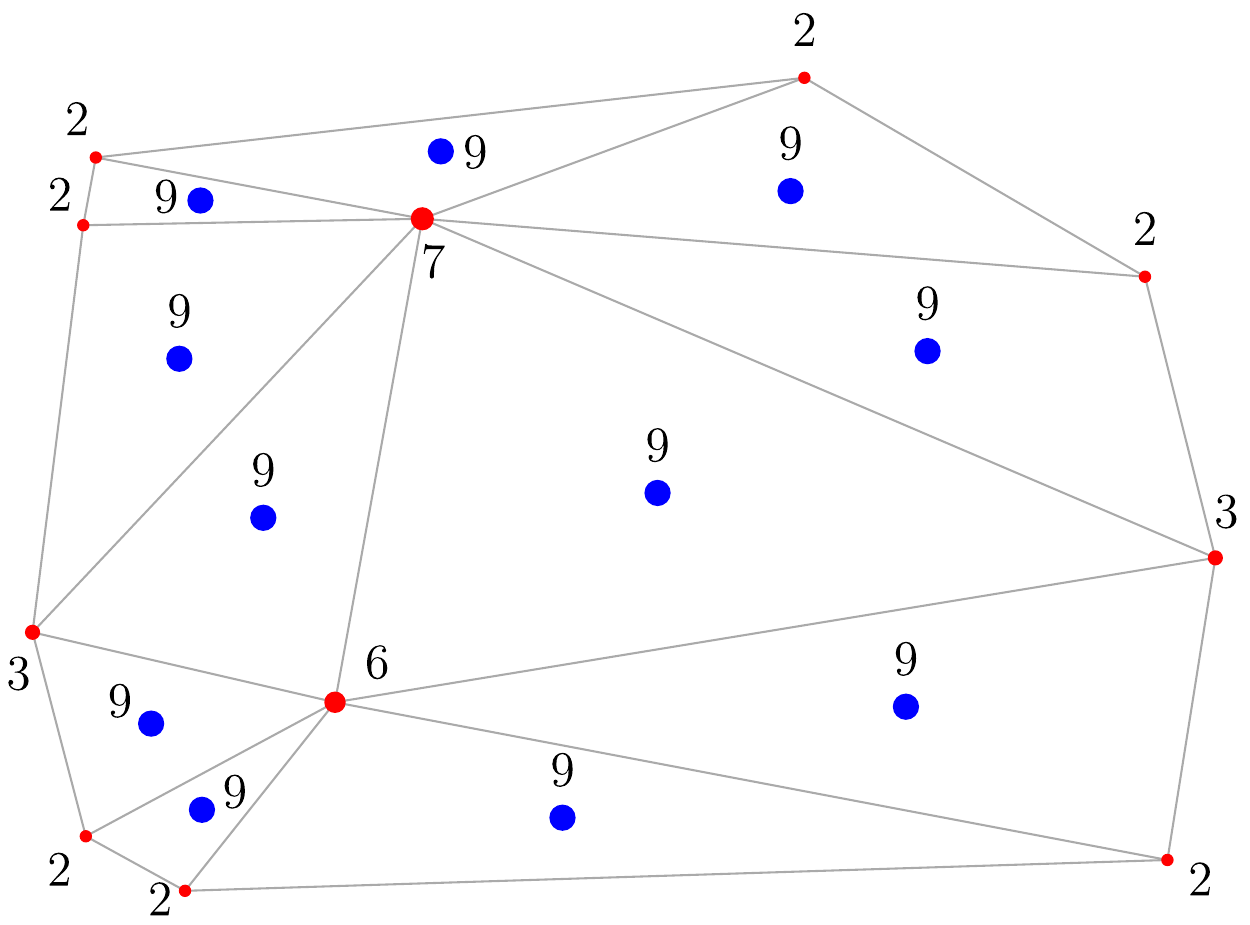}
        \caption{\textbf{2-design over a surface:} Example of a 2-design over the flat Lebesgue measure defined for a random polygon, based on its triangulation and the generalized  Simpson rule for $\Delta_3$ simplex, $1:9:1:1$. Note that the weights in mesh vertices (red) are equal to the number of simplices meeting at any given vertex, equal to the degree of a vertex in the bulk, while in the centroids (blue) we have constant weight of 9.}
        \label{fig:mesh_averaging}
    \end{figure}

    \break
    
    \section{Channel designs} \label{sec:channel}

    We will now consider pushforward designs induced by using partial trace as the pushforward map. In the case of the space of quantum states, the resulting pushforward designs, referred to as mixed state designs, have already been considered in~\cite{czartowski2020isoentangled}. Any pure-state design in a bipartite space $\mathcal{H}_{dk} = \mathcal{H}_d\otimes\mathcal{H}_k$ induces, therefore, a design in the space of mixed states.

    Let us now move on to the concept of channel designs as a new example of pushforward design. Consider a set of channels $S = \qty{\Phi_i}_{i=1}^m$ acting on a $d$-dimensional Hilbert space $\mathcal{H}_d$. We define a channel $[t,k]$-design in a natural way, based on the environmental form of a quantum channel, in relation to the average of the action of $t$ copies of a unitary operation $U$ over the set of unitary operations $\mathcal{U}(dk)$ of dimension $dk$ with a $k$-dimensional environment initialised in a pure state $\ket{0}_E\in\mathcal{H}_k$,
        \begin{align}
            \forall_{\rho\in\mathcal{B}(\mathcal{H}^{d^t})} \frac{1}{m} \sum_{i=1}^m \Phi^{\otimes t}_i(\rho) & = \int_{\mathcal{U}(dk)} \Tr_E\qty[U^{\otimes t}\qty(\rho\otimes\op{0}^{\otimes t}_{E})\qty(U^{\dagger})^{\otimes t}] \dd{U} \nonumber\\
            & \equiv \ev{\Tr_E\qty[U^{\otimes t}\qty(\rho\otimes\op{0}^{\otimes t}_{E})\qty(U^{\dagger})^{\otimes t}]}_{\mathcal{U}(dk)},     
        \end{align}
    where in the second line, for brevity, we introduce the notation for averaging over the unitary group $\mathcal{U}(dk)$. In other words, for any state $\rho$ of a $d^t$-dimensional system, the average effect of applying $t$-fold copies of channels from a channel $[t,k]$-design $S$ should be indistinguishable from the average effect of applying random unitaries of dimension $dk$ to the system, extended by a $k^t$-dimensional environment. A direct consequence of the above is that a unitary $t$-design in dimension $dk$ induces, by partial trace, a channel $[t,k]$-design in dimension $d$.  In particular, $[t,d^2]$-design represents channel design with respect to the flat (Lebesgue) measure    which will be also shortly denoted as channel $t$-design.
    
    One can also reformulate the statement in terms of corresponding Choi-Jamiołkowski states and the measure induced on the space of Choi-Jamiołkowski states, 
    \begin{widetext}
    \begin{equation}\label{eq:choi_average}
        \frac{1}{m} \sum_{i=1}^m \sigma_{\Phi_i^{\otimes t}} = 
        \ev{\Tr_E\qty[\qty(U_{\vb{AE}}^{\otimes t}\otimes\mathbb{I}_{\vb{B}}^{\otimes t})\qty(\op{\psi_+}_{\vb{AB}}^{\otimes t}\otimes\op{0}_{\vb{E}}^{\otimes t})\qty(\qty(U^{\dagger}_{\vb{AE}})^{\otimes t}\otimes\mathbb{I}_{\vb{B}}^{\otimes t})]}_{\mathcal{U}(dk)}
    \end{equation}
    \end{widetext}
    where the bold letters $\vb{A} \equiv A_1A_2\hdots A_t$ and similarly $\vb{B}$ and $\vb{E}$ correspond to $t$ copies of the principal systems $A,B$ and the environment $E$.

    Now, let us take a closer look at the RHS of \eqref{eq:choi_average} and first evaluate it explicitly for $t=1$. We may consider the block elements of the averaged Choi-Jamiołkowski state \mbox{$\ev{\sigma_{\Phi}}_{\mathcal{U}(dk)}\equiv \ev{\sigma_{\Phi}}$}, which are defined by
    \begin{equation}
        \ev{\sigma_\Phi}^a_b = \frac{1}{d}\Phi(\op{a}{b}) = \frac{1}{d}\sum_{l=0}^{k-1} K^{(l)}\op{a}{b}K^{(l)\dagger}.
    \end{equation}
    Shifting to the Stinespring dilation, we know that the elements of Kraus operators are defined explicitly as $(K^{(l)})^i_j = \tensor*{U}{*^l_0^i_j}$, which then allows us to consider the elements of the Choi-Jamiołkowski state as
    \begin{equation}
        \tensor*{\ev{\sigma_\Phi}}{*^a_b^i_j} = 
            \frac{1}{d}\ev{\sum_{l=0}^{k-1}
            \tensor*{U}{*^l_0^i_n}\delta^n_a\delta^b_m\tensor*{\qty(U^\dagger)}{*_l^0_j^m}}= 
            \frac{1}{d}\sum_{l=0}^{k-1}\ev{\tensor*{U}{*^l_0^i_a}\tensor*{\overline{U}}{*^l_0^j_b}}
    \end{equation}
    where we use the Einstein summation convention, with sum over $l$ index kept explicitly to highlight that it spans a different range of values; we also shortened $\ev{\cdot} \equiv\ev{\cdot}_{\mathcal{U}(dk)}$. In this case, due to symmetry, we may conclude without resorting to more sophisticated methods that
    \begin{equation}
        \tensor*{\ev{\sigma_\Phi}}{*^a_b^i_j} = 
            \frac{1}{d}\sum_{l=0}^{k-1} \frac{\delta^a_b\delta^i_j}{D} = \frac{\tensor*{\delta}{*^a_b^i_j}}{d^2}
    \end{equation}
    where for later convenience we introduced $D = dk$. This expression is in accordance with intuition that no specific direction should be distinguished in the space of states and thus, the averaged channel symmetric with respect to the adjoint action of the local unitary group $\mathcal{U}(d)$ should contract to the only point that does not change under this action, ie. the maximally mixed state. Thus, it cannot be anything other than the maximally depolarising channel. This points in fact to the following theorem.
    \begin{obs}
        channel $[1,k]$-design $\qty{\Phi_i}_{i=1}^m$ is also channel $[1,k']$-design for any $k'$.
    \end{obs}

    Expressions for $t=2$ become more involved and read 
    \begin{widetext}
        \begin{align}
            \tensor*{\ev{\sigma_{\Phi^{\otimes 2}}}}{*^{a_1}_{b_1}^{a_2}_{b_2}^{i_1}_{j_1}^{i_2}_{j_2}} 
                = & 
                \frac{1}{d^2}\ev{\sum_{l_1,l_2=0}^{k-1}
                \tensor*{U}{*^{l_1}_0^{i_1}_{n_1}}
                \tensor*{U}{*^{l_2}_0^{i_2}_{n_2}}
                \delta^{n_1}_{a_1}
                \delta^{n_2}_{a_2}
                \delta^{b_1}_{m_1}
                \delta^{b_2}_{m_2}
                \tensor*{\qty(U^\dagger)}{*_{l_1}^0_{j_1}^{m_1}}
                \tensor*{\qty(U^\dagger)}{*_{l_2}^0_{j_2}^{m_2}}}
                = \frac{1}{d^2}\sum_{l_1,l_2=0}^{k-1}\ev{
                \tensor*{U}{*^{l_1}_0^{i_1}_{a_1}}
                \tensor*{U}{*^{l_2}_0^{i_2}_{a_2}}
                \tensor*{\overline{U}}{*^{l_1}_0^{j_1}_{b_1}}
                \tensor*{\overline{U}}{*^{l_2}_0^{j_2}_{b_2}}} \nonumber\\
                \overset{(a)}{=} & \frac{1}{d^2} \sum_{l_1,l_2=0}^{k-1}\left[ \frac{1}{D^2-1} \qty(
                \tensor*{\delta}{*^{a_1}_{b_1}^{a_2}_{b_2}^{i_1}_{j_1}^{i_2}_{j_2}} + 
                \tensor*{\delta}{*^{a_1}_{b_2}^{a_2}_{b_1}^{i_1}_{j_2}^{i_2}_{j_1}^{l_1}_{l_2}})
                - \frac{1}{D(D^2-1)}\qty(
                \tensor*{\delta}{*^{a_1}_{b_2}^{a_2}_{b_1}^{i_1}_{j_1}^{i_2}_{j_2}} + 
                \tensor*{\delta}{*^{a_1}_{b_1}^{a_2}_{b_2}^{i_1}_{j_2}^{i_2}_{j_1}^{l_1}_{l_2}})\right] \\
                = & \frac{1}{d^2(D^2-1)}\qty(
                k^2 \tensor*{\delta}{*^{a_1}_{b_1}^{a_2}_{b_2}^{i_1}_{j_1}^{i_2}_{j_2}} +
                k   \tensor*{\delta}{*^{a_1}_{b_2}^{a_2}_{b_1}^{i_1}_{j_2}^{i_2}_{j_1}} -
                \frac{k^2}{D} \tensor*{\delta}{*^{a_1}_{b_1}^{a_2}_{b_2}^{i_1}_{j_1}^{i_2}_{j_2}} - 
                \frac{k}{D}   \tensor*{\delta}{*^{a_1}_{b_2}^{a_2}_{b_1}^{i_1}_{j_2}^{i_2}_{j_1}}) \nonumber
        \end{align}
    \end{widetext}
    with the factors after $(a)$ related to the Weingarten functions $\operatorname{Wg}(\sigma,D)$, connected to averaging over unitary group by
    \begin{equation}
        \int_{\mathcal{U}(D)} \dd{U}\prod_{n=1}^t U^{i_n}_{j_n}\overline{U}^{i'_n}_{j'_n} = \sum_{\sigma,\tau\in\mathcal{S}_t} 
        \operatorname{Wg}\qty(\sigma\tau^{-1},D)\qty(\prod_{n=1}^t
        \delta^{i_n}_{i'_{\sigma(n)}}
        \delta^{j_n}_{j'_{\tau(n)}}).
    \end{equation}
    In particular, Weingarten functions depend only on the permutation cycle structure~\cite{collins2004integration, Collins_2022}.

    \begin{figure*}
        \centering
        \includegraphics[width=.75\linewidth]{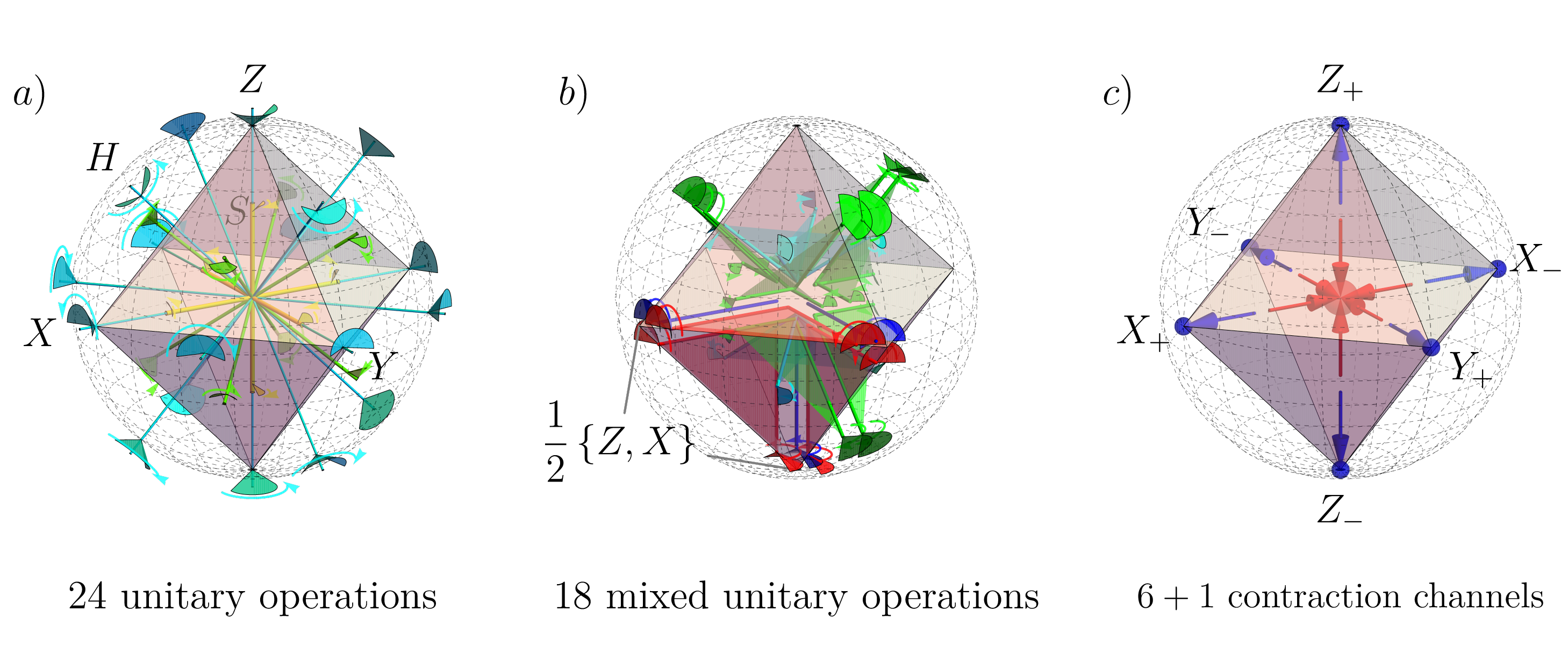}
        \caption{\textbf{Two-design in the set of single-qubit channels
        consists of 49 elements:} \textbf{a)} $24$ unitary operations comprising the Clifford group $\mathcal{C}_1$, each operation represented by the related axis and sector of a circle representing the angle of the rotation. Operations $\mathcal{R}_2$ of Kraus rank~2 include: \textbf{b)} 18 mixed unitary operations with equal-weight between different Clifford unitaries, each operation represented as a pair of semiaxes, with
        colours representing different mixtures; blue semiaxes represent mixing between half-rotations and identity transformation); \textbf{c)} 6 contraction  channels taking all the states into one of six vertices of the octahedron, each blue arrow represents an operation. The last 49-th operation -  the maximally depolarising channel has Kraus rank four and is represented by the circle at the center of the octahedron in panel c).
        }
        \label{fig:channel_design_comps}
    \end{figure*}

    By inspecting the above calculation we note that permutation between block indices $a_n,b_n$ is not equivalent to permutation for in-block indices $i_n,j_n$ -- a factor of $k^2$ appears when no permutation is effected, while only $k$ appears, due to evaluation of $\delta^{l_1}_{l_2}$, when there is a SWAP operation. Taking this into account, we give below the general formula
    \begin{equation} \label{eq:average_choi_explicit}
        \tensor*{\ev{\sigma_{\Phi^{\otimes t}}}}{*^{a_1}_{b_1}^{\hdots}_{\hdots}^{a_t}_{b_t}^{i_1}_{j_1}^{\hdots}_{\hdots}^{i_t}_{j_t}} = \sum_{\sigma,\tau\in\mathcal{S}_t} 
        \operatorname{Wg}\qty(\sigma\tau^{-1},D) k^{\operatorname{Cl}(\tau)}\qty(\prod_{n=1}^t
        \delta^{a_n}_{b_{\sigma(n)}}
        \delta^{i_n}_{j_{\tau(n)}})
    \end{equation}
    where $\operatorname{Cl}(\tau)$ is the total number of independent cycles in the permutation $\tau$.

    Using the above, one may formulate the following theorem.
    \begin{prop} \label{prop:analytical_criterion_channel_designs}
        Consider a set of channels $X = \qty{\Phi_i}$ acting on $d$-dimensional Hilbert space $\mathcal{H}^d$ and any $t\geq1$. For any set $X$ it holds that 
        \begin{equation}
            \norm{\frac{1}{\abs{X}}\sum_{i=1}^{\abs{X}} \sigma_{\Phi_i^{\otimes t}} - \ev{\sigma_{\Phi^{\otimes t}}}_{\mathcal{U}(dk)}} \geq 0
        \end{equation}
        with $\ev{\sigma_{\Phi^{\otimes t}}}_{\mathcal{U}(dk)}$ given explicitly in \eqref{eq:average_choi_explicit}. The equality is achieved if and only if the set $X$ is a channel $[t,k]$-design.
    \end{prop}

    For channel $[3,2]$-design on qubits, induced by the Clifford group on 2 qubits, we find 24 distinct unitary channels with multiplicity 96 each and 24 channels with Kraus rank 2 and multiplicity 384. These lead to a unit weight for each unitary and weight $4$ for the rank-2 operations. The entire channel $[3,2]$-design, consisting of 48 single-qubit channels, is visually depicted in Fig.~\ref{fig:channel_design_comps}.

    First, we note that the 24 unitary operations that form the set $\mathcal{R}_1$ correspond naturally to the elements of the Clifford group on a single qubit, $\mathcal{R}_1 = \mathcal{C}_1$.

    The remaining 24 channels of rank 2, collected in the set $\mathcal{R}_2$, are given by Choi-Jamiołkowski states with eigenvalues $\qty{1/2,1/2,0,0}$. They can be easily summarised by pairs of product states that form the support of the operators,
    {\small
    \begin{equation}
        \mathcal{R}_2 = \qty{\sigma_\Phi \!=\! \frac{1}{2}\sum_{i=0}^1\op{\psi_i}:\ket{\psi_i} = \qty(A\otimes B)\operatorname{CNOT}_{2\rightarrow1}^{k}\qty(\sigma_x^{j}\otimes\mathbb{I})\ket{0i}}  \label{eq:rd_2_channels}  
    \end{equation}
    }
    with $j,k\in\qty{0,1}$ and $A,B\in\qty{\mathbb{I},H,SH}$, after discarding duplicate channels. In particular, the subset for $k=1$ consists only of complete contraction channels towards 6 vertices of an octahedron inscribed in the Bloch sphere. The other 18 operations of Kraus rank 2 are equal combinations of pairs of Clifford unitaries.

    An example of a channel $[3,4]$-design is hard to calculate explicitly due to the order of Clifford group on 3 qubits, which is equal to $92\,897\,280$. Thus, it is necessary to resort to a sampling approach, using up to $5\cdot10^6$ sample channels sampled from a flat distribution over the Clifford group using a method demonstrated in~\cite{Berg} 
    and later verification using agreement of the average operator as given in \eqref{eq:average_choi_explicit}. This leads to the following observations.
    \begin{enumerate}
        \item The induced design is composed of the same channels as the channel $[3,4]$-design induced by Clifford group on 2 qubits, extended by maximally depolarising channel, corresponding to the maximally mixed Choi-Jamiołkowski state, \mbox{$\sigma_\Phi = \mathbb{I}_4/4$}.
        \item The weights are estimated to be $1$ for each rank-1 channel, $12$ for each rank-2 channel, and $192$ for the maximally depolarising channel. The weights are confirmed analytically using Proposition \ref{prop:analytical_criterion_channel_designs}.
        \item No rank-3 channels are necessary to construct a channel $[2,4]$-design on two qubits.
    \end{enumerate}

    The above leads to the claim that a channel $[3,k]$-design for any $k\in(1,\infty]$ can be achieved by taking the channels of the Clifford group $\mathcal{C}_1$, the rank-2 channels $\mathcal{R}_2$, and the maximally depolarising channel, taken with $k$-dependent weights.
    This claim can be verified by considering the following average,
    \begin{equation} \label{eq:2-design_decompo_chan}
        \frac{1}{a+b+c}\qty(a\sum_{\Phi\in\mathcal{C}_1}\sigma_{\Phi^{\otimes 2}}+
        b\sum_{\Phi\in\mathcal{R}_2}\sigma_{\Phi^{\otimes 2}} + c \frac{\mathbb{I}_{16}}{16}) = \ev{\sigma_{\Phi^{\otimes 2}}}_{\mathcal{U}(2k)}.
    \end{equation}
    We may assume without loss of generality $a=1$ and by solving for $b$ and $c$ we find the weight corresponding to the rank-2 channels as $b = 4(k-1)$, and the weight of the maximally depolarising channel as \mbox{$c = 32(k^2-3k+2)$}. Furthermore, we know that for $k = 2^l$ the expressions must match the Clifford group on $l+1$ qubits, and thus $a,b,c$ are also proper weights for $3$ designs in such cases. 

    Finally, Clifford group on two qubits $\mathcal{C}_2$ gives rise to a 3-design in the space of unistochastic channels acting on a single qubit, according to the definition presented in Appendix~\ref{app:unistoch}. Such a design is in fact composed of $43$ channels, divided into $24$ unitary channels from the single-qubit Clifford group $\mathcal{C}_1$ with weight $1$ each, $18$ mixed unitary channels $\mathcal{R}'_2 = \eval{\mathcal{R}_2}_{\text{mixed unitary}}$ with a weight $12$ assigned to each and, finally, the maximally depolarising channel of weight $240$.

\section{Effective environment dimension estimation} \label{sec:dimension}

The elements of the average Choi-Jamiołkowski state, given in \eqref{eq:average_choi_explicit}, are polynomial in
dimension $k$ of the environment, 
which one can extend to admit also non-integer values.
This observation allows us to propose the concept of \textbf{effective environment dimension}\footnote{To avoid confusion we stress that the name refers to the dimension of the underlying Hilbert space rather the spatio-temporal 4 dimensions.} $k^*$, defined as
    \begin{equation} \label{eq:effective_dimension}
        k^* = \underset{k\geq1}{\operatorname{arg\,min}}\norm{\frac{1}{\abs{S}}\sum_{i=1}^{\abs{S}} \sigma_{\Phi_i^{\otimes t}} - \ev{\sigma_{\Phi^{\otimes t}}}_{\mathcal{U}(dk)}},
    \end{equation}
    where $S$ is representative set of channels in a given physical system. One possibility of applying this metric is in quantum computers. Assuming first that gate implementations are low noise, so that quantum process tomography can be implemented with fidelity high enough, one may consider allowing free evolution for time $T$ and then evaluating the effective dimension of the environment for the resulting noise channels $\Theta_T$. Such a method can give insight into the effective number of levels $k^*(T)$ that actually interact with the primary system after a given time.

    One may also argue that the minimum distance itself, 
    \begin{equation} \label{eq:effective_dimension_error}
      \epsilon^* = \min_{k\geq1}\norm{\frac{1}{\abs{S}}\sum_{i=1}^{\abs{S}} \sigma_{\Phi_i^{\otimes t}} - \ev{\sigma_{\Phi^{\otimes t}}}_{\mathcal{U}(dk)}}
    \end{equation}
    is an informative measure; if this norm does not go to zero, it can signify that the interaction with the environment is not uniform and, as such, may require further investigation and adjustment of the underlying measure from which random noise actually originates. 

    Unfortunately, estimation of $\frac{1}{\abs{S}}\sum_{i=1}^{\abs{S}} \sigma_{\Phi_i^{\otimes t}}$ is problematic, as it cannot be performed using standard ancilla-assisted process tomography~\cite{Scott_2008}. Let us recall that in order to prepare the
    Choi-Jamio{\l}kowski state $\sigma_{\mathcal{E}}$ for a channel $\mathcal{E}$, one must first prepare a maximally entangled state $\ket{\psi_+} = \frac{1}{\sqrt{d}}\sum_{i=0}^{d-1} \ket{ii}$. 
    This can be implemented in a relatively straightforward manner within quantum computers, with a limited error involved. However, the next step, involving implementation of $\mathcal{E}$ on the first subsystem and identity channel $\mathcal{I}$ on the other, turns out to be problematic. If we were to make a tomography of an arbitrary channel $\mathcal{E}$, it would be doable up to errors. In the specific case where we are trying to make tomography of the very channel $\mathcal{N}_T$ of noise, such an approach is impossible, since by keeping the working qubits idle to accumulate the noise over time $T$, the ancillary systems are also accumulating similar noise.

    Taking into account the problem described above, we use a tomography scheme which does not rely on an unperturbed ancillary system, as it involves state tomography for a set of at least $d^2$ linearly independent quantum states, such that reconstruction of the blocks $\mathcal{N}_T(\op{a}{b})$ of the Choi-Jamiołkowski state can be done in post-processing~\cite{chuang1997, poyatos1997}. In our approach, we will take the complete set of MUB on two qubits as the set of input states, keep the device in question in an idle state for a fixed amount of time $T$ and then perform state tomography using measurements based on the same set of MUB.

    \begin{figure}[h]
        \centering
        \includegraphics[width=\linewidth]{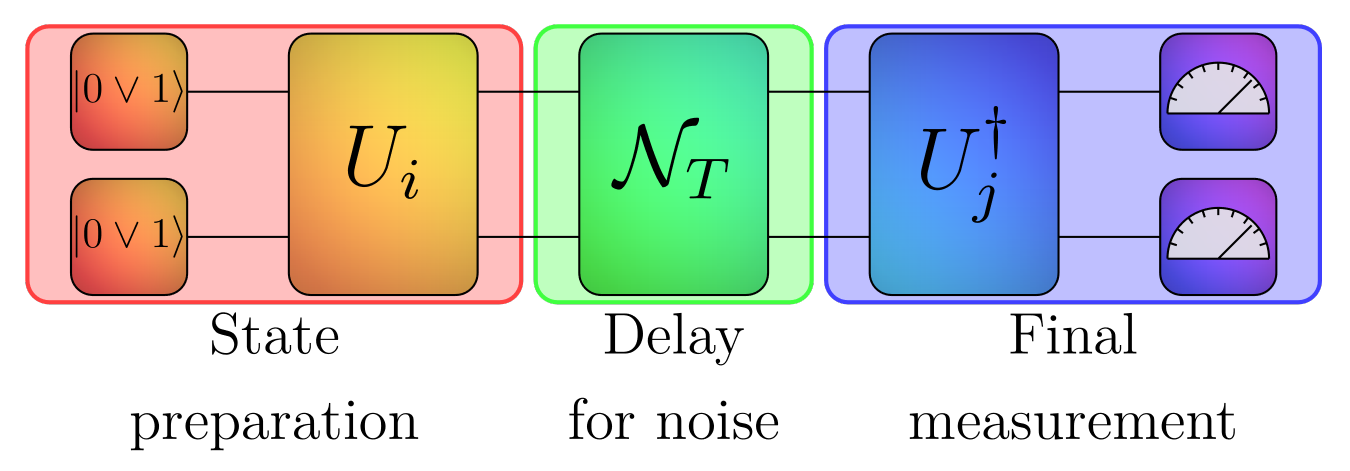}
        \caption{\textbf{Ancilla-free noise tomography scheme:} The circuit used for tomography of the noise channel $\mathcal{N}_T$ consists in three steps. First, one prepares the input state by preparing qubits in one of the computational basis states and acting on them with unitary matrix $U_i$ with $i\in\qty{0,\hdots,4}$, corresponding to one of five mutually unbiased bases (MUB). In the second step, the noise channel is implemented by forcing the quantum computer to delay the next operation by time $T$. In the third step, an inverse $U_j^\dagger$, $j\in\qty{0,\hdots,4}$ is implemented and both qubits are measured. 
        }
        \label{fig:tomograph_scheme}
    \end{figure}

    This procedure, depicted schematically in Fig.~\ref{fig:tomograph_scheme}, consists of three steps: 1) state preparation, 2) physical delay, acting as noise, and 3) final measurement. Details of the procedure are described in Appendix~\ref{app:tomo_scheme}.

    Using the above scheme, we have evaluated the effective dimension $k^*$ and the corresponding error $\epsilon^*$ for the 127-qubit quantum computer \textit{IBM Kyoto}, using 15 different times $T$ of interaction. The results of this demonstration are presented in Appendix~\ref{app:kyoto_plots}. As can be seen there, the effective dimension $k^*$ starts low for times on the order of $10$ $\mu$s, which are short relative to the energy relaxation time $t_1$ and the coherence time $t_2$, both of which are on the order of $100$ $\mu$s for \textit{IBM Kyoto}. For times two orders of magnitude larger, we observe a slow saturation of effective dimensionality around $k^* \approx 20$. Most importantly, the dimensionality obtained from the simulated noise closely matches that obtained from the actual quantum machine for both small and large times, while discrepancies appear at intermediate times. Importantly, the actual data does not exhibit the monotonic growth seen in the simulation, pointing either to unexpected effects at intermediate timescales or mismatch between the uniform interaction model and actual interactions with the environment.
    However, the error value $\epsilon^*$ shows that the uniform interaction model performs better for the actual machine than for the simulated noise, with a minimum that is not predicted by the simulation. For large times, the error $\epsilon^*$ seems to saturate at a nonzero value, $\eval{\epsilon^*}_{t\rightarrow\infty}\approx 0.3$.
    
    The non-vanishing $\epsilon^*$ and non-monotonic behaviour of $k^*$ as a function of interaction time suggests that uniform interaction with the environment may not be the most appropriate way to model the environment. One phenomenon not considered is that qubits are more prone to de-excitation, such as by emitting a photon to the environment, rather than excitation by photon absorption. To account for this, we modified the weight $4(k-1)$ from the equation~\eqref{eq:2-design_decompo_chan} assigned specifically to the complete emission channel to $4(k-1)+w$, adjusting the overall normalisation accordingly. 
    
    As shown in Appendix~\ref{app:kyoto_plots}, this noise simulation leads to the effective dimensionality of the environment around $k^* = 2$ for short times, increasing monotonically within the potential range of results as indicated by the simulated data to approximately $k^* \approx 2.2$ for long times, with agreement between simulation and demonstration. The additional weight $w$ assigned to the emission channel grows polynomially, as demonstrated by the linearity of the middle plot in the log-log scale. Finally, the $\epsilon^*$ value also exhibits an unpredicted minimum for relatively short times, and in the emission-related model, seems to converge to zero for long times, supporting the validity of the model in this regime.
    
    There are several takeaways from this analysis. First, uniform interaction of qubits in quantum computers is of limited application, with best performance for times on the order of the energy relaxation and decoherence times. Second, the emission model seems to provide a better fit, especially for longer times, with low effective dimensionality of the environment corresponding to the fact that IBM quantum computers are based on superconducting elements cooled to temperatures in the range of tens of millikelvins \cite{ibm2022goldeneye}. The final remark is that both models can be implemented relatively easily with the 49 channels contributing to Eq.~\eqref{eq:2-design_decompo_chan}, thus making it easy to simulate noise at a relatively low cost.

\section{Summary and Outlook} \label{sec:sum}

In this work we introduced the concept of pushforward designs, 
which link averaging sets
constructed for various spaces. 
In particular, 
any design consisting of pure quantum states
of size $d$ generates by decoherence
a design in the classical $d$-point
probability simplex. Another example is given by pure state
designs in a composite Hilbert
space, ${\cal H}_d \otimes {\cal H}_d$ which
lead, by partial trace, 
to designs in the set of mixed states
of size $d$. 

Focusing on averaging sets in a 
simplex induced from 
pure-state designs on complex projective space
$\mathbb{C}P^{d-1}$
we note generalized Simpson rules.
For dimension $d$ prime or power of prime
this construction follows from 
complete sets of mutually unbiased bases 
in ${\cal H}_d$. 
Analogous reasoning holds also for other dimensions
and
provides a simple method to approximate averages of functions over arbitrary hypersurfaces by applying simplex designs to triangulations.

The key result of this work consists
in an explicit construction of 
a family of $t$-designs
in the space of quantum channels. 
It allows one 
to average certain functions
over the space of quantum channels with respect to a measure induced by unitary operations on extended systems.
Taking the size $k$ of the environment
to be equal to $d^2$ one arrives
at a design corresponding
to the uniform, Lebesgue  measure in the
convex body of quantum channels.
We derive closed formulae for the $t$-copy Choi-Jamiołkowski states, thus providing a method for verification, whether a given set of channels actually provides a channel $[t,k]$-design. 

Moreover, we introduce the idea of effective environment dimensionality $k^*$, for which distance between theoretical and measured Choi-Jamiołkowski state related to noise is minimal. 
This idea was applied to estimate effective dimensionality of environment for the noise acting on qubits of an actual quantum computer, \textit{IBM Kyoto}. 
Our demonstration suggests that
the uniform distribution over quantum channels does not reflect 
the actual nature of the noise, but augmenting it by adjustable weight of emission allows us to estimate the effective dimension 
of the environment, $k^*\approx 2-2.2$,
for evolution time up to $350\mu\text{s}$.

Results introduced in this work can be extended in several 
directions. 
It is natural to 
extend the notion of exact designs
also for
 $\epsilon$-approximate $t$-designs
 and to obtain bounds on the accuracy $\epsilon$ in the image space~\cite{SCZ25}. Another 
 direction would be to 
 apply the notion 
 of pushforward designs for noncompact spaces~\cite{Markiewicz_2021}. Furthermore, it 
 is tempting 
 to refine application of
 effective environment dimensionality 
 as a benchmark of noise for quantum devices.

\begin{acknowledgments}
    It is a pleasure to thank Adam Sawicki and Marcin Markiewicz for useful comments and numerous discussions. 
    We acknowledge funding by the European Union under ERC 
    Advanced Grant \textit{TAtypic},  
    project number 101142236 and 
 by the National Science Centre, Poland,
under the contract No. 2021/03/Y/ST2/00193 within the QuantERA II Programme
that has received funding from the European Union’s Horizon 2020 research and innovation
programme under Grant Agreement No. 101017733. 
   
\end{acknowledgments}

\appendix

\section{Measures on the space of quantum channels}

As channel $[t,k]$-designs, approximating random quantum channels, play the key role in this work, for convenience of the readers, we provide here a short review of probability measures in the space of quantum channels. This appendix is based on \cite{KNPPZ21} and is included for self-containment of the manuscript.

The most general quantum operations 
        $\Phi:\Omega_d\mapsto \Omega_d$
        are given by completely positive trace-preserving (CPTP) maps. They can be defined in three equivalent manners:
        \begin{enumerate}
            \item By a set of $d_E$ Kraus operators $K_i$, with the action defined as
            \begin{equation}
                \Phi(\rho) = \sum_{i=1}^{d_E} K_i\rho K_i^\dagger
            \end{equation}
            with the identity resolution condition $\sum_i K_i^\dagger K_i = \mathbb{I}$ ensuring the trace preservation condition.
            \item By a unitary operation $U\in\mathcal{U}(d\times d_E)$ with the action of the channel defined by
            \begin{equation}
                \Phi(\rho) = \Tr_E\qty[U\qty(\rho\otimes\op{1})U^\dagger].
            \end{equation}
            \item By the corresponding Choi-Jamio{\l}kowski state $\sigma_\Phi$ defined as
            \begin{equation}
                \sigma_\Phi = \qty(\Phi\otimes\mathbb{I})\qty[\op{\psi_+}]
            \end{equation}
            where $\ket{\psi_+} = \frac{1}{\sqrt{d}} \sum_{i=1}^d \ket{ii}$ is the standard maximally entangled state.
        \end{enumerate}

    From this follow three natural ways for generating random quantum channels -- by random Kraus operators, random unitary operations, and random Choi-Jamio{\l}kowski states~\cite{KNPPZ21}. In order to understand the all the methods, we first need to shortly define a handful of ensembles of random matrices.

    \begin{defn}[Ginibre ensembles]
        A real Ginibre matrix $G_\mathbb{R}\in\mathbb{R}^{d_1\times d_2}$ that each element is taken from independent normal distribution $\mathcal{N}(1,0)$ of unit variance and zero mean. Complex Ginibre matrix $G_{\mathbb{C}}\equiv G = \qty(G'_\mathbb{R} + i G''_{\mathbb{R}})/\sqrt{2}$.
    \end{defn}
    
    \begin{defn}[Gaussian Unitary ensemble]
        Gaussian Unitary ensemble (GUE) of $d\times d$ Hermitian matrices, which is invariant under unitary operations, contains matrices of the form $H = (G + G^\dagger)/2$, where $G$ is a matrix taken from complex Ginibre ensemble.
    \end{defn} 
    
    \begin{defn}[Wishart ensemble]
        Complex Wishart ensemble with parameters $(d,s)$ consists of matrices of the form $W = GG^\dagger$, where $G$ is a rectangular $d\times s$ matrix taken from the complex Ginibre ensemble. 
    \end{defn}

    \begin{note}
        Wishart ensemble can be extended to allow for real parameter $s\in\qty{1,\hdots,d^2-1}\cup[d^2,\infty)\equiv \mathcal{S}$~
       ~\cite{PR91}.
    \end{note}

    \begin{defn}[Circular Unitary ensemble]
        Circular Unitary ensemble (CUE) consists of unitary matrices $U\in\mathcal{U}(d)$ distributed according to the Haar measure.
    \end{defn}

    Equipped with the above ingredients, we are ready to provide three recipes for generating random quantum channels, corresponding to the three  aforementioned (equivalent) definitions, taken from~\cite{BRUZDA2009320, KNPPZ21} and adapted for channels that preserve the dimensionality of the system, i.e. $\Phi:\mathcal{H}_d\mapsto\mathcal{H}_d$.
    
    \begin{const}[Random Quantum Channels via Kraus operators]\label{cons:Kraus}
        Let $s\geq1$ be an arbitrary natural number. We define $\mu^{\text{Kraus}}_{d,s}$ to be the probability measure of the randum quantum channel $\Phi$ defined in the following way:
        \begin{enumerate}
            \item Take $s$ random matrices $G_1,\hdots,G_s$ of size $d\times d$, drawn independently from the complex Ginibre ensemble.
            \item Compute the matrix $H = \sum_{i=1}^s G_i^\dagger G_i$.
            \item Define Kraus operators as $K_i = G_i H^{-1/2}$, $i\in\qty{1,\hdots,s}$.
        \end{enumerate}
    \end{const}

    \begin{const}[Random Quantum Channels via Choi-Jamiołkowski states]\label{cons:Choi}
        Let $s\in\mathcal{S}$ be a real number. We define $\mu^{\text{Choi}}_{d,s}$ to be the probability measure of the random quantum channel $\Phi$ defined in the following way:
        \begin{enumerate}
            \item Take a random complex Wishart matrix $W$ with parameters $(d^2,s)$;
            \item Find the positive semidefinite matrix defined by partial trace $H\equiv\Tr_B(W)$;
            \item Construct the sampled Choi-Jamiołkowski state as
            \begin{equation}
                \sigma_{\Phi} \equiv \qty(H^{-1/2}\otimes\mathbb{I})W\qty(H^{-1/2}\otimes\mathbb{I}).
            \end{equation}
        \end{enumerate}
    \end{const}

    \begin{const}[Random Quantum Channels via environmental form]\label{cons:env}
        Let $s$ be an integer and $s\geq1$. We define $\mu^{\text{Stinespring}}_{d,s}$ to be the probability measure of the random quantum channel $\Phi$, defined as follows:
        \begin{enumerate}
            \item Consider a random unitary matrix $U\in\mathcal{U}(Md)$, acting on the principal system and the environment $E$ of dimension $M$.
            \item The channel $\Phi$ is defined as
            \begin{equation}
                \Phi(\rho) = \Tr_E \qty[U\qty(\rho\otimes\op{0}_E)U^\dagger].
            \end{equation}
        \end{enumerate}
    \end{const}
    For Constructions \ref{cons:Choi} and \ref{cons:Kraus} one needs to take note of the zero-measure sets, in which the matrices $H$ are non-invertible and, therefore, do not adhere to the scheme.

    All the above constructions should be complemented with the natural (flat) Lebesgue measure $\mu^{\text{Lebesgue}}_d$ on the space of channels, induced by normalising the Hilbert-Schmidt volume to 1. In~\cite{KNPPZ21} the following equality of measures has been proven:
    \begin{equation}\label{eq:flat_meas}
        \mu^{\text{Lebesgue}}_d = \mu^{\text{Stinespring}}_{d,d^2} =\mu^{\text{Choi}}_{d,d^2} =\mu^{\text{Kraus}}_{d,d^2},
    \end{equation}
    with the first two equalities being crucial to our  considerations.

    \begin{widetext}
    \section{Isocoherent MUB}\label{app:isocoh_MUB}

    In this appendix we aim to find a global transformation $U$ a complete set of mutually unbiased bases $\qty{B_i}_{i=0}^4$ in dimension $d=4$ such that the classical states obtained by decoherence with respect to the computational basis is identical, ie. for $\ket{\psi}\in U B_i$  we have $\abs{\ip{\psi}{j}}^2 = p_{\sigma(j)}$ for some permutation $\sigma\in\mathcal{S}_4$ and a fixed set of probabilities $\qty{p_0,\,p_1,\,p_2,\,p_3}$. 
    Note that this notion of isocoherence stems from the resource theory of coherence, where the states are considered coherent or incoherent with respect to a given measurement basis \cite{winter2016operational}.
    Using a simple method of random walk on the unitary group $U(4)$ under cost function $\mathcal{C}$ reflecting the isocoherent property, 
    \begin{equation*}
        \mathcal{C}\qty(U) = \min_{\sum_j p_j = 1,p_j\geq 0}\qty(\sum_{u=0}^4\sum_{\ket{\psi}\in B_u} \min_{\sigma\in\mathcal{S}_4} \sum_{j=0}^3\abs{\abs{\mel{\psi}{U}{j}}^2 - p_{\sigma_j}})
    \end{equation*}
    and later simplification of the solution $\mathcal{C}(U) = 0$ by dephasing we found a particularly elegant form of isocoherent MUB given by 
    \begin{equation}
        B_j = \mqty(p_3&p_2&p_1&p_0\\p_0&p_3&p_2&p_1\\p_2&p_1&p_0&p_3\\p_1&p_0&p_3&p_2) \odot \operatorname{EXP}\qty[i \Phi_j]
    \end{equation}
    where the multiplication $\odot$ and exponentiation $\operatorname{EXP}$ are taken to be element-wise.
    The amplitudes are given by
        \begin{equation}
            \begin{aligned}
                p_0 & = \frac{1}{2} \sqrt{ 1+\frac{1}{\sqrt{5}}+\frac{1}{5}\sqrt{10+2 \sqrt{5}}},&
                p_1 & = \frac{1}{2} \sqrt{1-\frac{1}{\sqrt{5}}+\frac{1}{5}\sqrt{10-2 \sqrt{5}}},\\
                p_2 & = \frac{1}{2} \sqrt{1+\frac{1}{\sqrt{5}}-\frac{1}{5} \sqrt{10+2 \sqrt{5}}}, &
                p_3 & = \frac{1}{2} \sqrt{1-\frac{1}{\sqrt{5}}-\frac{1}{5}\sqrt{10-2 \sqrt{5}}},
            \end{aligned}    
        \end{equation}
    and phase matrices $\qty{\Phi_j}_{j=0}^4$ are given as 
        \begin{equation}
            \begin{aligned}
                \Phi_0 & = \mqty(
         0 & 0 & 0 & 0 \\
         0 & 0 & 0 & \pi  \\
         0 & 0 & \pi  & 0 \\
         0 & \pi  & \pi  & 0 \\), & 
                \Phi_1 & = \mqty(
         0 & 0 & 0 & 0 \\
         -\theta_{--} & -\theta_{--} & -\theta_{--} & \theta_{++} \\
         \theta_{-+} & \theta_{-+} & -\theta_{+-} & \theta_{-+} \\
         -\theta_{-+} & \theta_{+-} & \theta_{+-} & -\theta_{-+} \\
                ), &
                \Phi_2 & = \mqty(
         0 & 0 & 0 & 0 \\
         \theta_{--} & \theta_{--} & \theta_{--} & -\theta_{++} \\
         -\theta_{-+} & -\theta_{-+} & \theta_{+-} & -\theta_{-+} \\
         \theta_{-+} & -\theta_{+-} & -\theta_{+-} & \theta_{-+} \\
                ), \\&&
                \Phi_3 & = \mqty(
         0 & 0 & 0 & 0 \\
         -\theta_{-+} & -\theta_{-+} & -\theta_{-+} & \theta_{+-} \\
         -\theta_{--} & -\theta_{--} & \theta_{++} & -\theta_{--} \\
         \theta_{--} & -\theta_{++} & -\theta_{++} & \theta_{--} \\
                ), & 
                \Phi_4 & = \mqty(
         0 & 0 & 0 & 0 \\
         \theta_{-+} & \theta_{-+} & \theta_{-+} & -\theta_{+-} \\
         \theta_{--} & \theta_{--} & -\theta_{++} & \theta_{--} \\
         -\theta_{--} & \theta_{++} & \theta_{++} & -\theta_{--} \\
                ),
            \end{aligned}
        \end{equation}
    with $\theta_{\pm\pm} = \arccos\left[\frac{1}{4} \left(\pm1\pm\sqrt{5}\right)\right]$. This provides an example of a set of isocoherent MUB which might not be generated by powers of a single unitary operation. Since the amplitudes for all states are identical up to permutation, the corresponding linear entropy, a polynomial of order 2, for each state is equal to $\sum p_i^2 = 2/5$, which is fixed by the 2-design property \cite{Czartowski_2018}. More intriguingly, purity of the reduced states attains only one of two values, $\Tr(\Tr_B(\rho_i)^2) = 4/5 \pm 1/5\sqrt{5}$, with half of the states corresponding to each. We note that this property is similar to the one obtained by Zhu, Teo and Englert for two-qubit SIC-POVM \cite{Zhu2010twoqubit}.
    \end{widetext}

    \section{Unistochastic designs} \label{app:unistoch}

    In this appendix we will present an extension of the concept of channel designs, given in the main body of the manuscript, to unistochastic channels, defined by the action of a unitary operation $U\in\mathcal{U}(d^2)$ on a $d$-dimensional state $\rho$ extended by a maximally mixed state,
    \begin{equation}
        \Xi_U(\rho) = \Tr_E\qty[U\qty(\rho\otimes\frac{\mathbb{I}_E}{d})U^\dagger].
    \end{equation}
    The name is derived from the fact that both unistochastic matrices and unistochastic channels are defined uniquely by the underlying unitary operation \cite{bengtsson2004onduality}.
    
    With this definition one can see that a unitary $t$-design $X$ on dimension $d^2$ induces a unistochastic $t$-design,
    \begin{equation}
        \frac{1}{\abs{X}}\sum_{U_i\in X} \Xi_{U_i}^{\otimes t}\qty(\rho) = \int_{\mathcal{U}(d^2)} \Tr_E\qty[U^{\otimes t}\qty(\rho\otimes\frac{\mathbb{I}_E}{d^t})U^{\dagger\otimes t}] \dd{U}.
    \end{equation}
    where the integration is performed over the measure induced by the Haar measure of $U(d^2)$. This measure is not equivalent to the Lebesgue measure in the space of unistochastic channels \cite{musz2013unitary}, but it is straightforward  to  generate random operations with respect to it.

    We may carry out similar calculations in this case as for the channel designs,

    \begin{widetext}
    \begin{align}
        \tensor*{\ev{\sigma_{\Xi^{\otimes 2}}}}{*^{a_1}_{b_1}^{a_2}_{b_2}^{i_1}_{j_1}^{i_2}_{j_2}} 
            = & 
            \frac{1}{d^4}\ev{\sum_{l_1,l_2,o_1,o_2=0}^{d-1}
            \tensor*{U}{*^{l_1}_{o_1}^{i_1}_{n_1}}
            \tensor*{U}{*^{l_2}_{o_2}^{i_2}_{n_2}}
            \delta^{n_1}_{a_1}
            \delta^{n_2}_{a_2}
            \delta^{b_1}_{m_1}
            \delta^{b_2}_{m_2}
            \tensor*{\qty(U^\dagger)}{*_{l_1}^{o_1}_{j_1}^{m_1}}
            \tensor*{\qty(U^\dagger)}{*_{l_2}^{o_2}_{j_2}^{m_2}}} \\
            = & \frac{1}{d^4}\sum_{l_1,l_2,o_1,o_2=0}^{d-1}\ev{
            \tensor*{U}{*^{l_1}_{o_1}^{i_1}_{a_1}}
            \tensor*{U}{*^{l_2}_{o_2}^{i_2}_{a_2}}
            \tensor*{\overline{U}}{*^{l_1}_{o_1}^{j_1}_{b_1}}
            \tensor*{\overline{U}}{*^{l_2}_{o_2}^{j_2}_{b_2}}} \\
            \overset{(a)}{=} & \frac{1}{d^4} \sum_{l_1,l_2,o_1,o_2=0}^{d-1}\left[ \frac{1}{D^2-1} \qty(
            \tensor*{\delta}{*^{a_1}_{b_1}^{a_2}_{b_2}^{i_1}_{j_1}^{i_2}_{j_2}} + 
            \tensor*{\delta}{*^{a_1}_{b_2}^{a_2}_{b_1}^{i_1}_{j_2}^{i_2}_{j_1}^{l_1}_{l_2}^{o_1}_{o_2}})
            - \frac{1}{D(D^2-1)}\qty(
            \tensor*{\delta}{*^{a_1}_{b_2}^{a_2}_{b_1}^{i_1}_{j_1}^{i_2}_{j_2}} + 
            \tensor*{\delta}{*^{a_1}_{b_1}^{a_2}_{b_2}^{i_1}_{j_2}^{i_2}_{j_1}^{l_1}_{l_2}^{o_1}_{o_2}})\right] \label{eq:first_D_leftover}\\
            = & \frac{1}{d^4(D^2-1)}\qty(
            d^4 \tensor*{\delta}{*^{a_1}_{b_1}^{a_2}_{b_2}^{i_1}_{j_1}^{i_2}_{j_2}} +
            d^2   \tensor*{\delta}{*^{a_1}_{b_2}^{a_2}_{b_1}^{i_1}_{j_2}^{i_2}_{j_1}} -
            \frac{d^4}{D} \tensor*{\delta}{*^{a_1}_{b_1}^{a_2}_{b_2}^{i_1}_{j_1}^{i_2}_{j_2}} - 
            \frac{d^2}{D}   \tensor*{\delta}{*^{a_1}_{b_2}^{a_2}_{b_1}^{i_1}_{j_2}^{i_2}_{j_1}}) \label{eq:second_D_leftover}\\
            = & \frac{1}{d^4-1} \qty(\tensor*{\delta}{*^{a_1}_{b_1}^{a_2}_{b_2}^{i_1}_{j_1}^{i_2}_{j_2}} 
            + \frac{1}{d^2} \qty(\tensor*{\delta}{*^{a_1}_{b_2}^{a_2}_{b_1}^{i_1}_{j_2}^{i_2}_{j_1}}
            - \tensor*{\delta}{*^{a_1}_{b_1}^{a_2}_{b_2}^{i_1}_{j_1}^{i_2}_{j_2}})
            - \frac{1}{d^4}\tensor*{\delta}{*^{a_1}_{b_2}^{a_2}_{b_1}^{i_1}_{j_2}^{i_2}_{j_1}})
    \end{align}
    \end{widetext}
    where we kept $D = d^2$ in \eqref{eq:first_D_leftover} and \eqref{eq:second_D_leftover} to facilitate comparison between the expressions obtained here and for channel $[t,k]$-designs.

    Note that the resulting average is different from the formula for channel designs, and generalizing it for arbitrary $t$ yields
    \begin{equation} \label{eq:average_choi_explicit_unistoch}
        \tensor*{\ev{\sigma_{\Xi^{\otimes t}}}}{*^{a_1}_{b_1}^{\hdots}_{\hdots}^{a_t}_{b_t}^{i_1}_{j_1}^{\hdots}_{\hdots}^{i_t}_{j_t}} =  
        \sum_{\sigma,\tau\in\mathcal{S}_t}
        \operatorname{Wg}\qty(\sigma\tau^{-1},d^2) d^{2\operatorname{Cl}(\tau)}\qty(\prod_{n=1}^t
        \delta^{a_n}_{b_{\sigma(n)}}
        \delta^{i_n}_{j_{\tau(n)}}).
    \end{equation}
    This expression can be used to verify whether any given set of operations yields a proper design on unistochastic channels induced by unitary operations of dimension $d^2$.

    \section{Tomographic scheme} \label{app:tomo_scheme}

    In this appendix, we will describe in detail the tomographic scheme implemented for noise estimation noise in the \textit{IBM Kyoto}, as depicted schematically in Fig.~\ref{fig:tomograph_scheme}.
    
    In the first step we select a state from the computational basis and act on it with a unitary $U_i$ preparing the $i$-th MUB, and thus we have the state
    \begin{equation}
        \ket{\psi_1} = U_i \ket{\psi_0},\quad\ket{\psi_0}\in\qty{\ket{00},\,\ket{01},\,\ket{10},\,\ket{11}},
    \end{equation}
    and the matrix $U_i$ coming explicitly from the set
    \begin{align*}
        U_0 &= \mathbb{I}\\
        U_1 &= H^{\otimes 2} &
        U_2 &= S^{\otimes 2} U_1 \\
        U_3 &= \text{CNOT}_{1\rightarrow2}\text{CNOT}_{2\rightarrow1}U_2 &
        U_4 &= \text{CNOT}_{2\rightarrow1}\text{CNOT}_{1\rightarrow2}U_2
    \end{align*}
    with 
    \begin{align*}
        H & = \frac{1}{\sqrt{2}}\mqty(1&1\\1&-1),&
        \text{CNOT}_{1\rightarrow2} & = \mqty(1&0&0&0\\0&1&0&0\\0&0&0&1\\0&0&1&0),\\
        S & = \mqty(1 & 0 \\ 0 & i),&
        \text{CNOT}_{2\rightarrow1} & = \mqty(1&0&0&0\\0&0&0&1\\0&0&1&0\\0&1&0&0).
    \end{align*}

    After the second step, the state becomes, in principle, mixed,
    \begin{equation}
        \rho_2 = \mathcal{N}_T\qty(\op{\psi_1}).
    \end{equation}

    The final step uses the inverses $U_j^\dagger$ of the MUB preparations, so the measured state is
    \begin{equation}
        \rho_3 = U_j^\dagger \rho_2 U_j.
    \end{equation}

    The data obtained from this procedure are then fed back to the formula \eqref{eq:proj_des_reco_form} to reconstruct the states $\rho_2$ for each preparation, producing a total of $20\cdot 5 = 100$ circuits per time $T$.

    Finally, we may consider the full set of states $\rho_2$, which for no noise ($T=0$)should be given by all the MUB states, $\qty{\op{e^j_i}}_{i,j=1}^{d,d+1}$ with $d=4$ for the case of two qubits. We may consider a matrix $\vb{A}$ defined row-wise by vectorizations of the MUB states. We may use its Moore-Penrose pseudoinverse $\vb{A}^+$ to obtain
    \begin{equation}
         \eval{\qty(\vb{A}^+ \vb{A})}_{T=0} = \mathbb{I},
    \end{equation}
    which, under the same interpretation, consists of row-wise vectorizations of all projectors $\op{a}{b}$ with $a,b\in\qty{1,\hdots,d}$. 
    Naturally, the states $\rho_2$ are affected by noise acting in a non-zero time $T>0$. Thus, evaluating the same expression for the reconstructed matrix $\vb{A},$
    \begin{equation}
        \qty(\eval{\vb{A}^+}_{T=0})\qty(\eval{\vb{A}}_{T>0})
    \end{equation}
    gives row-wise vectorizations of the images of the projectors under noise, $\mathcal{N}_T(\op{a}{b})$, from which one may reconstruct the Choi-Jamiołkowski state $$\sigma_{\mathcal{N}_T} = \sum_{a,b=1}^4 \mathcal{N}_T(\op{a}{b})\otimes \op{a}{b}.$$
\begin{widetext}
\section{Effective dimension fits for \textit{IBM Kyoto} quantum computer} \label{app:kyoto_plots}

 Using the described tomography scheme described in Appendix \ref{app:tomo_scheme}, we have evaluated the effective dimension $k^*$ and the corresponding $\epsilon^*$ for the qubits 9 and 10 of the 127-qubit quantum computer \textit{IBM Kyoto} (Eagle architecture, see Fig.~\ref{fig:kyoto_architecture}), using 15 points 
    spread evenly from 10 $\mu$s to approximately 500 $\mu$ with respect to the logarithmic scale. Each point corresponds to 100 circuits with 1000 shots each, thus resulting in a total of $10^5$ measurements per delay time. Furthermore, we have used IBM quantum simulators with the noise data from \textit{IBM Kyoto} at the time of actual measurements to estimate the potential region within which the real-world data should fall. The results have been summarised in Figs.~\ref{fig:dimension_pure_model} and \ref{fig:dimension_emission_model}.

    \begin{figure}[h]
        \centering
        \includegraphics[width=0.5
        \linewidth]{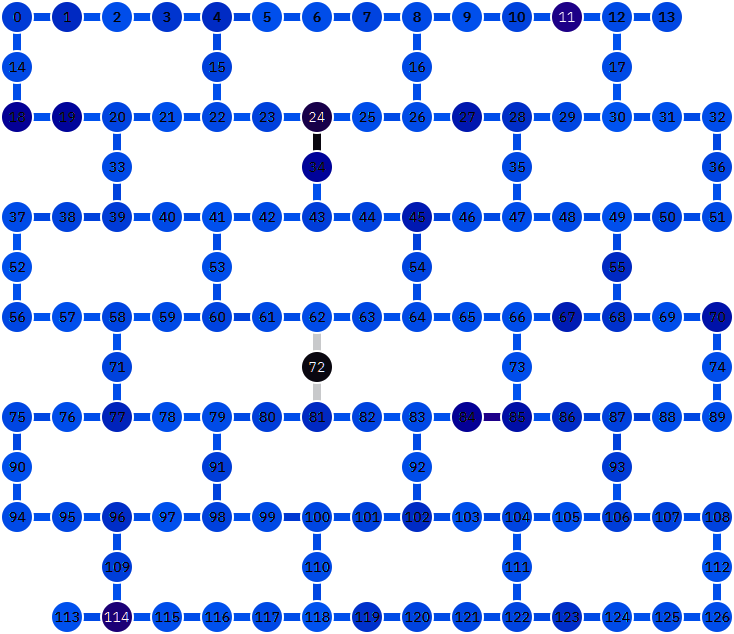}
        \caption{Schematic depiction of the architecture of 127-qubit \textit{IBM Kyoto} quantum computer. Callibration data for the computer at the time of measurements made for the demonstration are available at \cite{callibration_data}}
        \label{fig:kyoto_architecture}
    \end{figure}
    \begin{figure}[H]
        \centering
        \includegraphics[width=.9\linewidth]{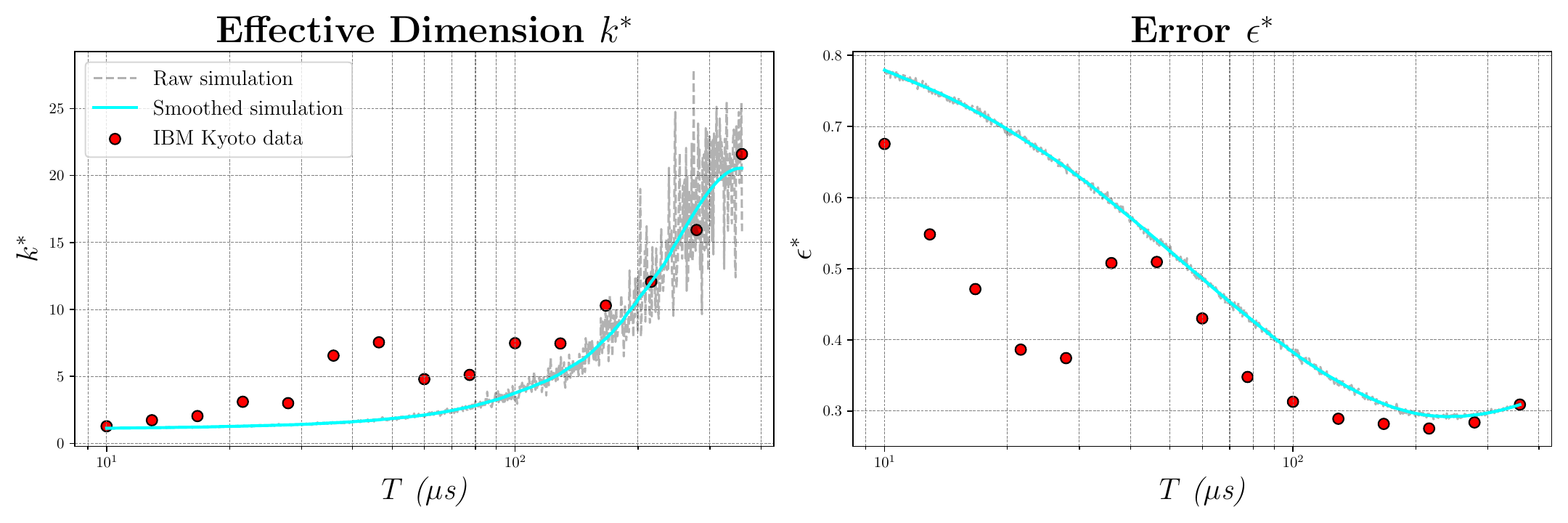}
        \caption{
        \textbf{Effective environment dimensionality for \textit{IBM Kyoto} quantum computer:}
            The effective dimension $k^*$ and corresponding error $\epsilon^*$ for qubits 9 and 10 were estimated over varying noise times $T$. Simulations based on the computer noise model (translucent black for raw data, cyan for smoothed) and data from the physical machine (red points) were utilised.
            In the left plot, $k^*$ increases from approximately 1 for short times to approximately 20 for longer times, indicating the effective dimensionality over time.
            The right plot illustrates $\epsilon^*$, showing a better fit for the actual machine data compared with the simulations, especially noticeable at shorter times. Error saturation is observed for longer timescales. 
        }
        \label{fig:dimension_pure_model}
    \end{figure}

    \begin{figure}[H]
        \centering
        \includegraphics[width=.9\linewidth]{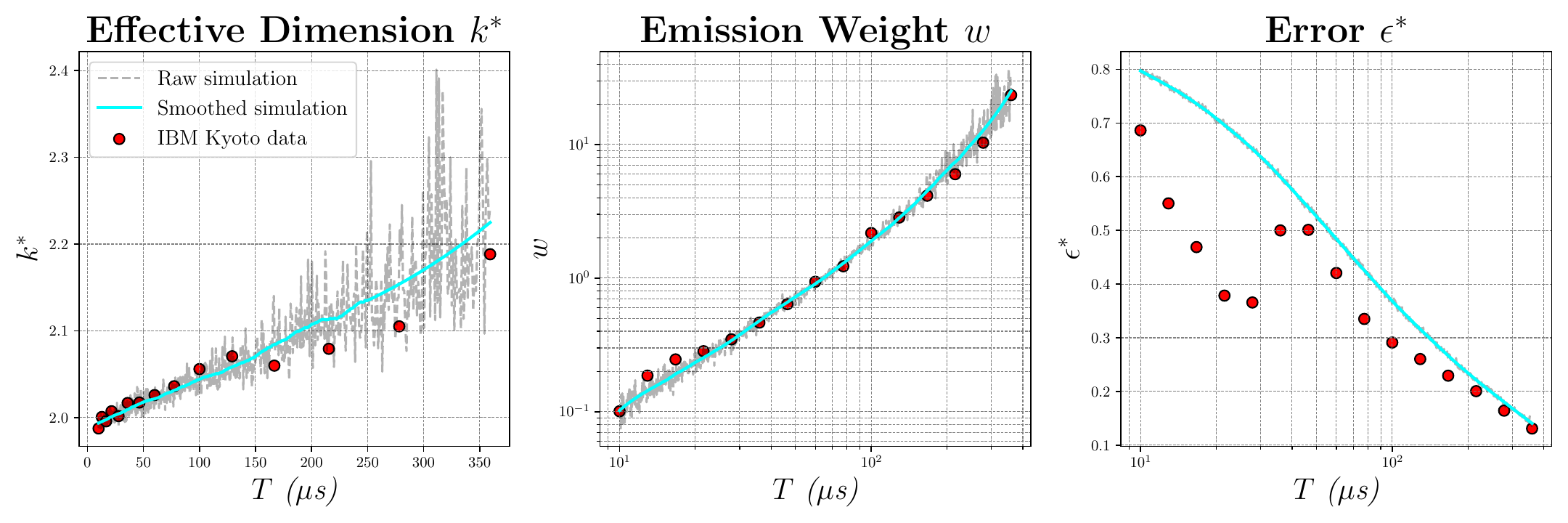}
        \caption{
        \textbf{Effective environment dimensionality with enhanced emission for \textit{IBM Kyoto} quantum computer:}
            The data from Fig.~\ref{fig:dimension_pure_model} is utilized to estimate the dimensionality $k^*$ of the environment under enhanced emission conditions, characterized by an additional weight $w$.
            In the middle plot, the added weight $w$ is depicted, while the right plot shows the estimated dimension $k^*$, which linearly transitions from 2 to 2.2 over time $T$.
            The right plot demonstrates that the model with enhanced emission provides a better fit to actual machine data compared with simulations, particularly noticeable as time $T$ increases. The error appears to vanish for longer timescales.
        }
        \label{fig:dimension_emission_model}
    \end{figure}
\end{widetext}

\vfill

\bibliography{references}

\end{document}